\documentclass[
    sigconf,
    screen,
    nonacm,
]{acmart}

\setcopyright{none}

\usepackage{mathtools}
\usepackage{stmaryrd}
\usepackage{algorithmic}
\usepackage{graphicx}
\usepackage{textcomp}
\usepackage{xcolor}
\def\BibTeX{{\rm B\kern-.05em{\sc i\kern-.025em b}\kern-.08em
    T\kern-.1667em\lower.7ex\hbox{E}\kern-.125emX}}

\usepackage{cmll}

\usepackage[all]{xy}
\usepackage{tikzit}

\usepackage{amsthm}
\usepackage{aliascnt}
\usepackage{cleveref} 

\theoremstyle{plain}
\newtheorem{theorem}{Theorem}

\newaliascnt{lemma}{theorem}
\newtheorem{lemma}[lemma]{Lemma}
\aliascntresetthe{lemma}

\newaliascnt{corollary}{theorem}
\newtheorem{corollary}[corollary]{Corollary}
\aliascntresetthe{corollary}

\newaliascnt{proposition}{theorem}
\newtheorem{proposition}[proposition]{Proposition}
\aliascntresetthe{proposition}

\newaliascnt{conjecture}{theorem}

\aliascntresetthe{conjecture}

\newtheorem*{theorem*}{Theorem}
\newtheorem*{proposition*}{Proposition}
\newtheorem*{claim*}{Claim}

\theoremstyle{definition}
\newaliascnt{definition}{theorem}
\newtheorem{definition}[definition]{Definition}
\aliascntresetthe{definition}

\newaliascnt{example}{theorem}
\newtheorem{example}[example]{Example}
\aliascntresetthe{example}

\newaliascnt{question}{theorem}

\aliascntresetthe{question}

\newtheorem*{problem*}{Problem}
\newtheorem*{question*}{Question}

\theoremstyle{remark}
\newaliascnt{remark}{theorem}
\newtheorem{remark}[remark]{Remark}
\aliascntresetthe{remark}

\newaliascnt{notation}{theorem}
\newtheorem{notation}[notation]{Notation}
\aliascntresetthe{notation}

\newtheorem*{notation*}{Notation}

\crefname{theorem}{Theorem}{Theorems}
\crefname{lemma}{Lemma}{Lemmas}
\crefname{corollary}{Corollary}{Corollaries}
\crefname{proposition}{Proposition}{Propositions}
\crefname{conjecture}{Conjecture}{Conjectures}

\crefname{definition}{Definition}{Definitions}
\crefname{example}{Example}{Examples}
\crefname{question}{Question}{Questions}

\crefname{remark}{Remark}{Remarks}

\crefname{section}{Section}{Sections}

\crefname{figure}{Figure}{Figures}

\definecolor{DarkGreen}{RGB}{0,120,0}
\definecolor{DarkBlue}{RGB}{0,0,200}
\definecolor{Cyan}{RGB}{0,139,139}

\usepackage{proof}
\newcommand{\subgroup}{\le}

\newcommand{\PropVar}{X}
\newcommand{\PropVarb}{Y}
\newcommand{\PropVarc}{Z}

\newcommand{\GroupoidCat}{\mathbf{Grpd}}

\newcommand{\creed}{\mathcal{A}}
\newcommand{\creedb}{\mathcal{B}}

\newcommand{\Grp}{\mathbb{A}}
\newcommand{\Grpb}{\mathbb{B}}
\newcommand{\Grpc}{\mathbb{C}}

\newcommand{\kit}{\mathcal{A}}

\newcommand{\InducedProf}[1]{\widetilde{#1}}

\newcommand\lohide[1]{}

\newcommand{\Powerset}{P}

\newcommand{\Int}{\mathbb{Z}}
\newcommand{\defeq}{:=}
\newcommand{\defiff}{:\Longleftrightarrow}
\newcommand{\orthogonal}{\mathrel{\bot}}

\newcommand{\CreedOf}{\mathfrak{C}}
\newcommand{\KitOf}{\mathfrak{K}}

\newcommand{\Profunctor}{\mathcal{P}}
\newcommand{\Profunctorb}{\mathcal{Q}}
\newcommand{\Profunctorc}{\mathcal{R}}
\newcommand{\Group}{G}
\newcommand{\Groupb}{H}
\newcommand{\Groupc}{K}

\newcommand{\Formula}{A}
\newcommand{\Formulab}{B}
\newcommand{\Formulac}{C}

\newcommand{\TotalElement}{\mathit{Tot}}
\newcommand{\CototalElement}{\mathit{Cot}}

\newcommand{\SProf}{\mathbf{SProf}}

\newcommand{\TUnit}{\mathcal{I}}
\newcommand{\TGrp}{\creed}
\newcommand{\TGrpb}{\creedb}

\newcommand{\Underlying}{\mathcal{U}}

\newcommand{\Tot}{\mathbf{Tot}}

\newcommand{\Conjugate}{\mathrel{\equiv_{\mathit{conj}}}}

\newcommand{\ProfCat}{\mathbf{Prof}}

\newcommand{\TotProf}{\mathbf{STProf}}

\newcommand{\Truncation}[1]{|{#1}|}
\newcommand{\Stabiliser}{\mathsf{fix}}

\newcommand{\llpar}{\invamp}

\newcommand{\Functor}{\mathcal{F}}
\newcommand{\Functorb}{\mathcal{G}}

\newcommand{\profarrow}{\mathrel{\relbar\joinrel\mapstochar\rightarrow}}

\newcommand{\ident}{\mathrm{id}}
\newcommand{\sem}[1]{\llbracket{#1}\rrbracket}
\newcommand{\op}{\mathrm{op}}

\newcommand{\Set}{\mathbf{Set}}

\newcommand{\Rel}{\mathbf{Rel}}

\newcommand{\Prof}{\mathbf{Prof}}

\begin{document}

\title{Full Definability in a Profunctorial Model}

\author{Takeshi Tsukada}
\email{t.tsukada@acm.org}
\affiliation{\institution{Chiba University}
  \country{Japan}
}
\orcid{0000-0002-2824-8708}

\author{Kazuyuki Asada}
\affiliation{\institution{Tohoku University}
  \country{Japan}
}
\orcid{0000-0001-8782-2119}

\author{Kengo Hirata}
\affiliation{\institution{University of Edinburgh}
  \country{UK}
}
\affiliation{\institution{Kyoto University}
  \country{Japan}
}
\orcid{0009-0005-4416-2655}

\begin{abstract}
    A semantic model enjoys full definability if every semantic element in the model is a denotation of some proof or program.
    Full definability indicates that the model captures programs and proofs in a highly detailed manner.

    This paper studies full definability in a model based on the (bi)category of profunctors on groupoids, which is a proof-relevant variant of the relational model.
Despite the fact that a profunctor is far more complicated than a relation, we show that a rather straightforward application of the ideas for the relational model, together with the notion of stability in profunctors, provides a complete characterisation of definable profunctors.
    More precisely, all \emph{logical} families of \emph{stable} and \emph{total} profunctors are definable by proof-nets of multiplicative linear logic with MIX.
    As a part of the full definability proof, we show that the stability serves as a correctness criterion, which we think is of independent interest. 
\end{abstract}

\maketitle

\section{Introduction}
The relational model based on the category of sets and relations is one of the most fundamental models of linear logic~\cite{Girard1987}.
It is both concise and conceptually rich, and it connects naturally to a wide range of other notions, including intersection types~\cite{Carvalho2017} and a game model~\cite{Hyland1999}.
The relational model itself forms a compact closed category, so the tensor product \( \otimes \) and par \( \llpar \) collapse in the relational model.
However, by equipping sets with additional structure and restricting to relations that respect that structure, one can obtain many important models that are not necessarily compact closed.
Examples of such models include coherence spaces~\cite{Girard1987}, hypercoherences~\cite{Ehrhard1993}, finiteness spaces~\cite{Ehrhard2008}, and totality spaces~\cite{Loader1994}.

The (bi)category of profunctors can be seen as a categorification of the category of relations.
While a relation records only whether two elements are related or not, a profunctor carries richer information: it can encode how many ways two objects may be related, and even some relationships among those ways.
Owing to this additional structure, profunctor-based models are deeply connected to several important themes, including the Taylor expansion~\cite{Tsukada2017,Tsukada2018} (in the sense of Danos and Ehrhard~\cite{Danos2011}), type derivations in intersection type systems~\cite{Olimpieri2021}, and game semantics~\cite{Clairambault2023,Clairambault2024}.
For this reason, the (bi)category of profunctors has attracted considerable attention in recent years~\cite{Fiore2007,Tsukada2017,Tsukada2018,Olimpieri2021,Clairambault2023,Fiore2024}.

It is known that profunctors arising as interpretations of linear-logic proofs and of \( \lambda \)-terms enjoy special properties.
For instance, one can associate a matrix to a profunctor; although this passage is not functorial for arbitrary profunctors, it behaves as a functor when restricted to profunctors that are denotations of \( \lambda \)-terms~\cite{Tsukada2018,Clairambault2023}.
It is therefore natural to ask what intrinsic properties characterize profunctors that are representable by proofs or \( \lambda \)-terms.

This paper addresses this question and, as an answer, presents a fully definable model based on profunctors.

A model is said to be \emph{fully definable} (or \emph{fully complete}) if every element of the model arises as the denotation of some proof or \( \lambda \)-term.
A number of fully definable models have been constructed.
For example, Loader~\cite{Loader1994} proposed the \emph{totality space model}, which is based on the relational model and fully definable for \( \mathsf{MLL} + \mathsf{Mix} \);
Devarajan~et~al.~\cite{Devarajan1999} shows that a model based on Chu spaces is fully definable for \( \mathsf{MLL} \) (without \( \mathsf{Mix} \));
and fully definable game models have been developed for a wide range of programming languages and proof systems (see, e.g., \cite{Abramsky1994,Abramsky2000,Hyland2000,Mellies2005a}).

Our construction is based on two key ideas in the literature.
The first is Loader's \emph{totality}~\cite{Loader1994}, which was employed to obtain a fully definable model built on the relational model.
Since profunctors generalize relations, totality also plays a central role in our study of profunctors.
The second idea is the concept of \emph{stable profunctor}, introduced by Taylor~\cite{Taylor1989} and further developed by Fiore et~al.~\cite{Fiore2022,Fiore2024,Fiore2024a}.
The concept of stable profunctor was motivated by Berry's stable domain theory~\cite{Berry1978} and Girard's normal functor~\cite{Girard1988} and proposed as a categorification of the stable domain model.
Fortunately, these two ideas can both be understood as instances of \emph{orthogonality} in the sense of Hyland and Schalk~\cite{Hyland2003}.
We propose \emph{stably total orthogonality}, which combines these two notions, and construct a fully definable model by using this orthogonality.
Given that profunctors are significantly more intricate than relations, it is perhaps surprising that such a simple construction yields a fully definable model.

In contrast to the construction of the model, the proof of full definability is not straightforward.
Furthermore, we obtain some observations that may be of independent interest.

First, we prove that the stability of profunctors yields a correctness criterion (\cref{thm:creed-kit-correctness-without-cut}).
This result can be understood as the profunctorial analogue of Retor{\'e}'s result~\cite{Retore1997} that coherence of the relational interpretation provides a correctness criterion.
Retor{\'e}'s proof and ours share the same high level idea of the proof, which suggests a close relationship between coherence and stability; however, at present we do not have a precise correspondence between them in any formal sense.

Second, we observe that strict factorization systems naturally arise in the stably total model (\cref{prop:totality:strict-factorization-system}).
Strict factorization systems on groupoids have previously been used for combinatorial analyses of profunctor models and related semantics~\cite{Tsukada2018,Olimpieri2021,Clairambault2023,Clairambault2024}.
The appearance of strict factorization systems in this paper seems to have the same origin.
However, the way strict factorization systems arise here differs from earlier work in two respects.
First, in this paper, strict factorization systems emerge from the notions of stability and totality, which at first sight appear unrelated to strict factorization systems.
This provides a new perspective that further justifies the role of strict factorization systems in the study of profunctors.
Second, whereas existing approaches typically equip each object with a single strict factorization system, in this paper we associate multiple strict factorization systems with each object.
Moreover, part of the information of \( \mathsf{MLL} \) proofs is also encoded in the choice of which strict factorization system to adopt among these alternatives.

The technical contributions of this paper can be summarized as follows.
\begin{itemize}
    \item We construct a new model based on profunctors that is fully definable with respect to \( \mathsf{MLL}+\mathsf{Mix} \).  Our result shows that definable profunctors are characterized by using \emph{totality} and \emph{stability}.
    \item We show that stability provides a new semantic correctness criterion for \( \mathsf{MLL}+\mathsf{Mix} \).
    \item We show that strict factorization systems arise from stability and totality. This provides a new perspective on the relevance of strict factorization systems in the analysis of profunctors. 
\end{itemize}

\subsection{Related Work}
The definability result of this paper is inspired by Loader's approach~\cite{Loader1994}.
Other technically related studies include results based on Chu spaces~\cite{Devarajan1999}.
These constructions can be explained in terms of orthogonality~\cite{Hyland2003}.

The full definability (or its variant such as compact definability) is considered as a key property of game models, and many game models with full definability with respect to a variety of programming languages and proof systems have been proposed (see, e.g., \cite{Abramsky1994,Abramsky2000,Hyland2000,Mellies2005a}).
A motivation for this work comes from attempts to recast game semantics in terms of (pre)sheaves and profunctors~\cite{Winskel2013,Tsukada2015,Eberhart2017,Eberhart2018,Clairambault2023,Clairambault2024,Jacq2018}.
Profunctor models may be viewed as a simplified variant of game semantics, obtained by discarding several structures, and one of our interests is to understand how far game-semantic reasoning can be carried out in such a simplified setting.
From this perspective, it is a particularly intriguing question whether the approach developed in this paper can be extended to the setting that includes exponentials.

As work on profunctor models as extensions of the relational model, one may first mention Fiore~et~al.~\cite{Fiore2007} (although there is also earlier related work, such as Taylor~\cite{Taylor1989}).
Fiore~et~al.~\cite{Fiore2007} introduced, in the setting of profunctor models, a structure corresponding to the Kleisli category of the multiset comonad of the relational model, and showed that it carries a cartesian closed structure.
Profunctor models are widely used as a tool for analysing the ``quantitative'' or ``combinatorial'' information of programs and proofs~\cite{Tsukada2017,Tsukada2018,Clairambault2023,Clairambault2024,Olimpieri2021,Fiore2022,Fiore2024,Fiore2024a,Ong2017,Jacq2018}, which may often be lost in the relational model.

 \section{Preliminaries}
\label{sec:pre}
This section briefly reviews multiplicative linear logic (MLL) and profunctors on groupoids.
We also discuss the interpretation of MLL proof structure in the (bi)category of profunctors.

\subsection{Profunctors on Groupoids}
A \emph{(small) groupoid} \( \Grp \) is a (small) category in which every morphism is invertible.
Then \( \Grp(a,a) \) is a group for every object \( a \in \Grp \).
We write \( \Grp, \Grpb, \Grpc \) for groupoids, \( a,b,c \) for their objects, and \( \alpha,\beta,\gamma \) for morphisms.
Let \( \GroupoidCat \) be the category of small groupoids and functors.

\begin{notation}
    The composition of morphisms is written in diagrammatic order: the composite of \( \alpha \colon a \longrightarrow b \) and \( \beta \colon b \longrightarrow c \) is \( (\alpha; \beta) \colon a \longrightarrow c \), which we sometimes simply write as \( \alpha \beta \).
    We write \( I \) for the trivial category, which has one object \( \star \) and one morphism \( \ident_\star \).
For groups \( \Group \) and \( \Groupb \), we write \( \Group \subgroup \Groupb \) to mean that \( \Group \) is a subgroup of \( \Groupb \).
\end{notation}

For groupoids \( \Grp \) and \( \Grpb \), a \emph{profunctor} \( \Profunctor \colon \Grp \profarrow \Grpb \) is a functor \( \Profunctor \colon \Grp^{\op} \times \Grpb \longrightarrow \Set \).
This definition works for general categories \( \Grp \) and \( \Grpb \), but this paper focuses on profunctors between groupoids.
For \( x \in \Profunctor(a,b) \), \( \alpha \in \Grp(a', a) \) and \( \beta \in \Grpb(b,b') \), we write \( \alpha \cdot x \cdot \beta \) for \( \Profunctor(\alpha,\beta)(x) \in \Profunctor(a', b') \).
When \( \alpha = \ident \), we simply write \( \ident \cdot x \cdot \beta \) as \( x \cdot \beta \), and similarly for \( \alpha \cdot x \defeq \alpha \cdot x \cdot \ident \).

Profunctors \( \Profunctor \colon \Grp \profarrow \Grpb \) and \( \Profunctorb \colon \Grpb \profarrow \Grpc \) compose in a manner similar to relations:
\begin{equation*}
    (\Profunctor ; \Profunctorb)(a,c)
    \quad\defeq\quad
    \int^{b \in \Grpb} \Profunctor(a,b) \times \Profunctorb(b,c).
\end{equation*}
Here \( \int^{b \in \Grpb} \) is an operation called a \emph{coend}, which can intuitively be thought of as a kind of existential quantification.
Concretely, the composition is
\begin{equation*}
    \textstyle
    (\Profunctor; \Profunctorb)(a,c)
    \quad\defeq\quad
    \left(\coprod_{b \in \Grpb} \Profunctor(a,b) \times \Profunctorb(b,c)\right)/{\sim},
\end{equation*}
where \( \sim \) is an equivalence relation defined by
\begin{equation*}
    (x,y) \sim (x',y')
    \:\defiff\:
    (x', y') = (x \cdot \beta^{-1}, \beta \cdot y)
    \mbox{ for some \( \beta \).}
\end{equation*}
This relation \(\sim\) is also given by \( (x, \beta \cdot y) \sim (x \cdot \beta, y) \).

Unlike relations, the composition of profunctors is usually not strictly associative, but associative up to isomorphisms.
For this reason, the categorical structure in which morphisms are profunctors is often treated as a \emph{bicategory}.
However, within the scope of this paper, no issues arise from taking quotients of profunctors by isomorphisms.
Thus, we define \( \ProfCat \) as the category where objects are groupoids and morphisms are isomorphism classes of profunctors.\footnote{This construction has previously been pointed out as raising concerns about size issue.  A simple way to justify the quotient here is to assume a Grothendieck universe \( U \) and to use \( \ProfCat_U \) of the bicategory of \( U \)-small groupoids and \( \Set_U \)-valued profunctors.}
This allows us to directly use results in the \( 1 \)-category theory, such as the definition and coherence theorem of compact closed categories.
By abuse of notation, we shall write as if a morphism is a profunctor (instead of an equivalence class of profunctors).

The category of profunctors is a model of linear logic.
The interpretations of logical connectives are given as follows:
\begin{align*}
    \Grp \otimes \Grpb &\defeq \Grp \times \Grpb
    &
    \Grp \llpar \Grpb &\defeq \Grp \times \Grpb.
\end{align*}
The negation is the opposite category \( \Grp^{*} \defeq \Grp^{\op} \).
Since \( \otimes \) and \( \llpar \) coincide on \( \Prof \), it is a compact closed category.
The unit \( \eta_{\Grp} \colon I \profarrow \Grp \times \Grp^{\op} \) is \( \eta_{\Grp}(\star, (a,a')) = \Grp(a',a) \) and the counit \( \epsilon_{\Grp} \colon \Grp^{\op} \times \Grp \profarrow I \) is \( \epsilon_{\Grp}((a', a), \star) = \Grp(a,a') \).

A profunctor induces a relation.
For a groupoid \( \Grp \), let \( \Truncation{\Grp} \) be the set of the isomorphism classes of \( \Grp \).
Then a profunctor \( \Profunctor \colon \Grp \profarrow \Grpb \) induces a relation \( \Truncation{\Profunctor} \subseteq \Truncation{\Grp} \times \Truncation{\Grpb} \) defined by \( ([a]_{\cong}, [b]_{\cong}) \in \Truncation{\Profunctor} \defiff \Profunctor(a,b) \neq \emptyset \).
It is easy to see that \( \Truncation{-} \colon \ProfCat \longrightarrow \Rel \) is a functor.

Presheaves and profunctors on groupoids have a useful representation as a sum of atomic ones, as discussed in \cite{Fiore2024}.
We shall heavily use this representation, so we give a concrete definition.

For profunctors \( \Profunctor_i \colon \Grp \profarrow \Grpb \),
the profunctor \( \coprod_{i \in I} \Profunctor_i \colon \Grp \profarrow \Grpb \) is defined by \( (\coprod_{i \in I} \Profunctor_i)(a,b) \defeq \coprod_{i \in I} (\Profunctor_i(a,b)) \) on objects and \( (\coprod_{i \in I} \Profunctor_i)(\alpha, \beta)(\mathsf{in}_j(x)) \defeq \mathsf{in}_j(\Profunctor_j(\alpha,\beta)(x)) \) on morphisms, where \( \mathsf{in}_j \colon \Profunctor_j(a,b) \longrightarrow (\coprod_{i \in I} \Profunctor(a,b)) \) is the injection.
When \( I = \{ 1,\dots,n \} \) is a finite set, we also write \( \Profunctor_1 + \dots + \Profunctor_n \) for \( \coprod_i \Profunctor_i \).
An ``atom'' with respect to \( + \) is described by a group, and every profunctor is expressed as the sum of atoms.

For a subgroup \( \Group \le \Grp^{\op}(a,a) \times \Grpb(b,b) \), 
we write \( \InducedProf{\Group} \colon \Grp \profarrow \Grpb \) for the profunctor defined by:
\begin{align*}
    \InducedProf{\Group}(a', b')
    &\quad\defeq\quad
    \{\, \alpha \cdot \Group \cdot \beta \mid \alpha \in \Grp(a', a), \beta \in \Grpb(b, b') \,\}
    \\
    \InducedProf{\Group}(\alpha,\beta)
    &\quad\defeq\quad
    X \mapsto (\alpha \cdot X \cdot \beta)
\end{align*}
where for \( X \subseteq \Grp^{\op}(a, a') \times \Grpb(b, b') \), \( \alpha \in \Grp^{\op}(a', a'') \) and \( \beta \in \Grpb(b',b'') \), the subset \( (\alpha \cdot X \cdot \beta) \subseteq \Grp^{\op}(a, a'') \times \Grpb(b,b'') \) is defined by \( (\alpha \cdot X \cdot \beta) \defeq \{ (\alpha \alpha_0, \beta_0 \beta) \mid (\alpha_0, \beta_0) \in X \} \) (where the composition \( \alpha \alpha_0 \) is taken in \( \Grp \)).
This profunctor cannot be decomposed by \( + \) any further.

Every profunctor can be written as a sum of profunctors of the form \( \InducedProf{\Group} \).
Given a profunctor \( \Profunctor \colon \Grp \profarrow \Grpb \), let \( (x_i \in \Profunctor(a_i, b_i))_{i \in I} \) be a family such that (1) for every \( x \in \Profunctor(a,b) \), we have \( x = \alpha \cdot x_i \cdot \beta \) for some \( i \), \( \alpha \) and \( \beta \), and (2) \( x_j = \alpha \cdot x_i \cdot \beta \) implies \( i = j \).
Then
\begin{equation*}
    \textstyle
    \Profunctor
    \quad\cong\quad
    \coprod_{i \in I} \InducedProf{\Stabiliser(x_i)}
\end{equation*}
where \( \Stabiliser(x) \le \Grp^{\op}(a,a) \times \Grpb(b,b) \) for \( x \in \Profunctor(a,b) \) is the \emph{stabiliser subgroup} defined by 
\begin{equation*}
    \Stabiliser(x)
    \:\defeq\:
    \{\, (\alpha,\beta) \in \Grp^{\op}(a,a) \times \Grpb(b,b) \mid \alpha \cdot x \cdot \beta = x \,\}.
\end{equation*}
\begin{proposition}\label{prop:prof-decomposition}
    A profunctor \( \Profunctor \colon \Grp \profarrow \Grpb \) can be expressed as
    \begin{equation*}
        \textstyle
        \Profunctor
        \:\cong\:
        \coprod_{i \in I} \InducedProf{\Group_i},
        \qquad
        \forall i \in I.~
        \Group_i \le \Grp^{\op}(a_i,a_i) \times \Grpb(b_i,b_i)
    \end{equation*}
    for some family \( (\Group_i)_{i \in I} \).
\qed
\end{proposition}
There is some flexibility in the choice of the family \( (x_i)_{i \in I} \),
and correspondingly, there is also flexibility in the choice of \( (\Group_i)_{i \in I} \).
For subgroups \( \Group, \Groupb \le \Grp^{\op}(a,a) \times \Grpb(b,b) \), we say \( \Group \) is \emph{conjugate} to \( \Groupb \) if
\begin{equation*}
    \Groupb = \{\, (\alpha^{-1} g \alpha,\, \beta h \beta^{-1}) \mid (g,h) \in \Group \,\}
\end{equation*}
for some \( \alpha \in \Grp(a,a) \) and \( \beta \in \Grpb(b,b) \).
We write \( \Group \Conjugate \Groupb \) when \( \Group \) is conjugate to \( \Groupb \).
\begin{proposition}
    Let \( \Group, \Groupb \le \Grp^{\op}(a,a) \times \Grpb(b,b) \).
    Then \( \InducedProf{\Group} \cong \InducedProf{\Groupb} \) if and only if\/ \( \Group \) is conjugate to \( \Groupb \).
\qed
\end{proposition}

\subsection{MLL proofs and proof structures}
An \emph{\( \mathsf{MLL} \) formula}, or \emph{formula}, is given by:
\begin{equation*}
    \Formula,\Formulab,\Formulac
    \quad::=\quad
    \PropVar \mid \PropVar^{\bot} \mid \Formula \otimes \Formulab \mid \Formula \llpar \Formulab,
\end{equation*}
where \( \PropVar \) is a propositional variable and \( \PropVar^{\bot} \) is its negation.
For simplicity, we consider only formulas in negation normal form (i.e.~the negation \( ({-})^{\bot} \) is applicable to only propositional variables).
This paper studies the unit-free fragment, as in many other papers on this subject~\cite{Loader1994,Devarajan1999}.

A \emph{sequent} \( \Gamma \) is a finite sequence of formulas.
The concatenation of sequents \( \Gamma \) and \( \Delta \) is written as \( \Gamma, \Delta \).
The (one-sided) sequent calculus \( \mathsf{MLL+Mix} \) is defined by the rules in \cref{fig:sequent-calculus-mll}.
We focus on the cut-free proofs, so the cut rule is not allowed.
This is justified by the fact that the cut elimination process does not change the interpretation.
\begin{figure}[t]
    \begin{gather*}
        \dfrac{\mathstrut}{\vdash \PropVar, \PropVar^\bot}(\mathsf{Ax})
        \quad
        \dfrac{\vdash \Gamma, A \qquad \vdash \Delta, B}{\vdash \Gamma, \Delta, A \otimes B}(\otimes)
        \quad
        \dfrac{\vdash \Gamma, A, B}{\vdash \Gamma, A \llpar B}(\llpar)
        \\[5pt]
        \dfrac{\vdash \Gamma, A, B, \Delta}{\vdash \Gamma, B, A, \Delta}(\mathsf{Ex})
        \quad
        \dfrac{\vdash \Gamma \qquad \vdash \Delta}{\vdash \Gamma, \Delta}(\mathsf{Mix})
\end{gather*}
    \vspace{-3ex}
    \caption{The proof rules of the sequent calculus \( \mathsf{MLL+Mix} \).}
    \label{fig:sequent-calculus-mll}
\end{figure}

A \emph{proof structure} \( \pi \) of a sequent \( \Gamma \) is an involutive function on occurrences of propositional variables that maps a positive occurrence of \( \PropVar \) to its negative occurrence.
A proof structure can be presented as a graph, consisting of the syntax tree together with extra links connecting \( \PropVar \) and \( \PropVar^{\bot} \) (see \cref{fig:pre:proof-structure}).
These extra links are called \emph{axiom links}.
\begin{figure}[t]
    \begin{tikzpicture}[yscale=0.6]
	\begin{pgfonlayer}{nodelayer}
		\node [style=none] (00) at (0.5, -0.7) {};
		\node [style=none] (0) at (0.5,0) {$\llpar$};
		\node [style=none] (1) at (-1, 1) {$\otimes$};
		\node [style=none] (2) at (-2, 2) {$\llpar$};
		\node [style=none] (3) at (-2.5, 3) {$\PropVar_1$};
		\node [style=none] (4) at (-1.5, 3) {$\PropVar_2$};
		\node [style=none] (5) at (-0.5, 2) {$\PropVar_2^\bot$};
		\node [style=none] (6) at (2, 1) {$\llpar$};
		\node [style=none] (7) at (1, 2) {$\llpar$};
		\node [style=none] (8) at (0.5, 3) {$\PropVar_2$};
		\node [style=none] (9) at (1.5, 3) {$\PropVar_2^\bot$};
		\node [style=none] (10) at (2.5, 2) {$\PropVar_1^\bot$};
		\node [style=none] (11) at (-2.5, 4.5) {};
		\node [style=none] (12) at (-1.5, 4) {};
		\node [style=none] (13) at (-0.5, 3.5) {};
		\node [style=none] (14) at (0.5, 3.5) {};
		\node [style=none] (15) at (1.5, 4) {};
		\node [style=none] (16) at (2.5, 4.5) {};
	\end{pgfonlayer}
	\begin{pgfonlayer}{edgelayer}
		\draw (00) to (0);
		\draw (0) to (1);
		\draw (1) to (2);
		\draw (1) to (5);
		\draw (2) to (3);
		\draw (2) to (4);
		\draw (0) to (6);
		\draw (6) to (7);
		\draw (6) to (10);
		\draw (7) to (8);
		\draw (7) to (9);
		\draw (3) to (11.center);
		\draw (4) to (12.center);
		\draw (5) to (13.center);
		\draw (8) to (14.center);
		\draw (9) to (15.center);
		\draw (10) to (16.center);
		\draw (11.center) to (16.center);
		\draw (12.center) to (15.center);
		\draw (13.center) to (14.center);
	\end{pgfonlayer}
\end{tikzpicture}
     \caption{A proof structure of \( ((\PropVar_1 \llpar \PropVar_2) \otimes \PropVar_2^{\bot}) \llpar ((\PropVar_2 \llpar \PropVar_2^\bot) \llpar \PropVar_1^{\bot}) \).}
    \label{fig:pre:proof-structure}
\end{figure}

\begin{figure}[t]
    \begin{tikzpicture}[yscale=0.6]
	\begin{pgfonlayer}{nodelayer}
		\node [style=none] (00) at (0.5, -0.7) {};
		\node [style=none] (0) at (0.5,0) {$\llpar$};
		\node [style=none] (1) at (-1, 1) {$\otimes$};
		\node [style=none] (2) at (-2, 2) {$\llpar$};
		\node [style=none] (3) at (-2.5, 3) {$\PropVar_1$};
		\node [style=none] (4) at (-1.5, 3) {$\PropVar_2$};
		\node [style=none] (5) at (-0.5, 2) {$\PropVar_2^\bot$};
		\node [style=none] (6) at (2, 1) {$\llpar$};
		\node [style=none] (7) at (1, 2) {$\llpar$};
		\node [style=none] (8) at (0.5, 3) {$\PropVar_2$};
		\node [style=none] (9) at (1.5, 3) {$\PropVar_2^\bot$};
		\node [style=none] (10) at (2.5, 2) {$\PropVar_1^\bot$};
		\node [style=none] (11) at (-2.5, 4) {};
		\node [style=none] (12) at (-1.5, 3.5) {};
		\node [style=none] (13) at (-0.5, 3.5) {};
		\node [style=none] (14) at (0.5, 3.5) {};
		\node [style=none] (15) at (1.5, 3.5) {};
		\node [style=none] (16) at (2.5, 4) {};
	\end{pgfonlayer}
	\begin{pgfonlayer}{edgelayer}
		\draw (00) to (0);
		\draw (0) to (1);
		\draw (1) to (2);
		\draw (1) to (5);
		\draw (2) to (3);
		\draw (2) to (4);
		\draw (0) to (6);
		\draw (6) to (7);
		\draw (6) to (10);
		\draw (7) to (8);
		\draw (7) to (9);
		\draw (3) to (11.center);
		\draw (4) to (12.center);
		\draw (5) to (13.center);
		\draw (8) to (14.center);
		\draw (9) to (15.center);
		\draw (10) to (16.center);
		\draw (11.center) to (16.center);
		\draw (12.center) to (13.center);
		\draw (14.center) to (15.center);
	\end{pgfonlayer}
\end{tikzpicture}
     \caption{An incorrect proof structure of \( ((\PropVar_1 \llpar \PropVar_2) \otimes \PropVar_2^{\bot}) \llpar ((\PropVar_2 \llpar \PropVar_2^\bot) \llpar \PropVar_1^{\bot}) \).  For a Danos-Regnier switching choosing the right branch for the left-most \( \llpar \), there is a cycle \( \llpar \)--\( X_2 \)--\( X_2^{\bot} \)--\( \otimes \)--\( \llpar \) on the left.}
    \label{fig:pre:incorrect-proof-structure}
\end{figure}

A sequent calculus proof induces a proof structure.
For example,
the following proof
\begin{equation*}
    \infer{\vdash ((\PropVar_1 \llpar \PropVar_2) \otimes \PropVar_2^{\bot}) \llpar ((\PropVar_2 \llpar \PropVar_2^\bot) \llpar \PropVar_1^{\bot})}{
        \infer{\vdash (\PropVar_1 \llpar \PropVar_2) \otimes \PropVar_2^{\bot}, \PropVar_2, \PropVar_2^\bot, \PropVar_1^{\bot}}{
            \infer{\vdash \PropVar_1 \llpar \PropVar_2, \PropVar_2^{\bot}, \PropVar_1^\bot}{
                \infer{\vdash \PropVar_1, \PropVar_2, \PropVar_2^{\bot}, \PropVar_1^\bot}{
                    \vdash \PropVar_1, \PropVar_1^\bot
                    &
                    \vdash \PropVar_2, \PropVar_2^\bot
                }
            }
            &
            \vdash \PropVar_2, \PropVar_2^\bot
        }
    }
\end{equation*}
induces the proof structure in \cref{fig:pre:proof-structure}.
However, a proof structure does not necessarily correspond to a proof.
An example of a proof structure with no corresponding proof is found in \cref{fig:pre:incorrect-proof-structure}.

A proof structure is \emph{sequentialisable} or \emph{correct} if it comes from a sequent calculus proof, and a correct proof structure is called a \emph{proof-net}.
A natural question asks when a proof structure is correct.
A condition characterizing the correctness is called a \emph{correctness criterion} and many correctness criteria have been proposed.
\emph{Danos-Regnier criterion} is a famous correctness criterion~\cite{Danos1989,Fleury1994}.
A \emph{Danos-Regnier switching} is a choice of left/right for each occurrence of \( \llpar \).
The graph associated to a Danos-Regnier switching is obtained by removing every edge between a \( \llpar \)-node and its unchosen child.
\begin{theorem}[Danos and Regnier~\cite{Danos1989}]\label{thm:danos-regnier}
    A proof structure is correct with respect to \( \mathsf{MLL} \) if and only if, for every Danos-Regnier switching, the associated graph is a tree.
\qed
\end{theorem}
\begin{theorem}[Fleury and Retor{\'e}~\cite{Fleury1994}]\label{thm:weak-danos-regnier}
    A proof structure is correct with respect to \( \mathsf{MLL+Mix} \) if and only if, for every Danos-Regnier switching, the associated graph is acyclic.
\qed
\end{theorem}

\subsection{Interpretation of Proof Structures}\label{sec:pre:interpretation}
It is well-known that an \( \mathsf{MLL} \) proof can be interpreted in any \( * \)-autonomous category, so a proof-net has an interpretation in any \(*\)-autonomous category.
However, interpreting a possibly incorrect proof structure needs an additional structure.
A sufficient structure to interpret a proof structure is a compact closed category, which is a \( * \)-autonomous category with \( \Formula \llpar \Formulab \cong \Formula \otimes \Formulab \), so it admits
\begin{equation*}
    \dfrac{\vdash \Gamma, \Formula, \Formulab}{\vdash \Gamma, \Formula \otimes \Formulab}(\otimes{=}\llpar).
\end{equation*}
In the (wrong) proof system \( \mathsf{MLL+Mix}+(\otimes{=}\llpar) \), every proof structure has a corresponding (possibly wrong) proof: Given a proof structure, prepare an instance of the axiom rule for each axiom link, take their MIX, and introduce logical connectives by using \( (\llpar) \), \( (\otimes{=}\llpar) \) and \( (\mathsf{Ex}) \).
By the coherence of compact closed categories, the interpretation of a proof in \( \mathsf{MLL+Mix}+(\otimes{=}\llpar) \) is determined by its proof structure.
Since \( \ProfCat \) is compact closed, a proof structure can be interpreted in \( \ProfCat \).

The interpretation of a formula is straightforward: it is obtained by replacing \( \otimes \) and \( \llpar \) with \( \times \) and \( ({-})^{\bot} \) with \( ({-})^{\op} \).
For example, the interpretation of \( ((\PropVar_1 \llpar \PropVar_2) \otimes \PropVar_2^{\bot}) \llpar ((\PropVar_2 \llpar \PropVar_2^\bot) \llpar \PropVar_1^{\bot}) \) is
\begin{equation*}
    ((\Grp_1 \times \Grp_2) \times \Grp_2^{\op}) \times ((\Grp_2 \times \Grp_2^\op) \times \Grp_1^{\op})
\end{equation*}
when the interpretation of \( \PropVar_i \) is \( \Grp_i \).
The interpretation of a sequent is similar.

The interpretation of a proof structure is slightly more involved but still we have two convenient ways to describe it.

The first one comes from the above description of the ``proof'' corresponding to a proof structure: it is the mix of axiom rules followed by exchange rules and introductions of logical connectives.
The axiom \( \vdash \PropVar^{\bot}, \PropVar \) is interpreted as the unit of the compact structure, which is \( \Grp({-}, {+}) \colon I \profarrow \Grp^{\op} \times \Grp \) in the case of \( \ProfCat \) (where \( \Grp \) is the interpretation of the propositional variable \( \alpha \)).
So the mix \( \vdash \PropVar_{\xi_1}^{\bot}, \PropVar_{\xi_1}, \dots, \PropVar_{\xi_n}^{\bot}, \PropVar_{\xi_n} \) of axiom rules gives a profunctor \( \Profunctor \colon I \profarrow \Grp_{\xi_1}^{\op} \times \Grp_{\xi_1} \times \dots \times \Grp_{\xi_n}^{\op} \times \Grp_{\xi_n} \) defined by \( \Profunctor(\star, (c_1, b_1, \dots, c_n, b_n)) \defeq \Grp_{\xi_1}(c_1, b_1)  \times \dots \times \Grp_{\xi_n}(c_n, b_n) \).
The interpretations of the exchange rule and the rules for logical connectives are just the symmetry and associativity laws of the tensor product \( \times \).
Hence, the interpretation of a proof structure is \( \Grp_{\xi_1}({-}, {+}) \times \dots \times \Grp_{\xi_n}({-}, {+}) \) followed by the symmetry and associativity.
\begin{example}
    Consider the proof structure in \cref{fig:pre:proof-structure}.
    The interpretation of the sequent is \( ((\Grp_1 \times \Grp_2) \times \Grp_2^{\op}) \times ((\Grp_2 \times \Grp_2^\op) \times \Grp_1^{\op}) \), and we have the canonical isomorphisms
    \begin{equation*}
        ((\mathbb{X} \times \mathbb{Y}) \times \mathbb{Z}) \times ((\mathbb{U} \times \mathbb{V}) \times \mathbb{W})
        \quad\cong\quad
        \mathbb{X} \times \mathbb{W} \times \mathbb{Y} \times \mathbb{V} \times \mathbb{Z} \times \mathbb{U}.
    \end{equation*}
    The interpretation of the proof structure in \cref{fig:pre:proof-structure} followed by the above isomorphism is
    \begin{equation*}
        \Grp_1({-}, {+}) \times \Grp_2({-}, {+}) \times \Grp_2({-}, {+}).
    \end{equation*}
A more detailed description of the interpretation of the proof structure is obtained by analysing the isomorphism
    \(
        ((\mathbb{X} \times \mathbb{Y}) \times \mathbb{Z}) \times ((\mathbb{U} \times \mathbb{V}) \times \mathbb{W})
        \quad\cong\quad
        \mathbb{W} \times \mathbb{X} \times \mathbb{V} \times \mathbb{Y} \times \mathbb{Z} \times \mathbb{U}
    \).
    This isomorphism in \( \ProfCat \) originates from an isomorphism in \( \GroupoidCat \).
    In general, a functor \( \Functor \colon \Grp \longrightarrow \Grpb \) induces profunctors \( \Functor_* \colon \Grp \profarrow \Grpb \) and \( \Functor^* \colon \Grpb \profarrow \Grp \) given by \( \Functor_*(a, b) \defeq \Grpb(F(a), b) \) and \( \Functor^*(b,a) \defeq \Grpb(b, F(a)) \).
    The symmetry and associativity laws in \( \Prof \) come from functors, so does the above isomorphism.
    So, the interpretation \( \chi \colon I \profarrow ((\Grp_1 \times \Grp_2) \times \Grp_2^{\op}) \times ((\Grp_2 \times \Grp_2^\op) \times \Grp_1^{\op}) \) is \( \chi(\star, (((x, y), z), ((u, v), w))) \cong \Grp_1(w,x) \times \Grp_2(v,y) \times \Grp_2(z,u) \).
\end{example}

The second one is based on annotations on proof structures in the spirit of \emph{experiments} in the relational interpretation~\cite{Girard1987}.
We call this scheme the \emph{profunctorial experiments}, which shall be explained in \cref{sec:experiments}.

 \section{Stability as a Correctness Criterion}
The notion of stability was introduced in the domain theory~\cite{Berry1978}, and Taylor~\cite{Taylor1989} and Fiore~et~al.~\cite{Fiore2022,Fiore2024,Fiore2024a} studied this notion in \( \ProfCat \) using \emph{creed} and \emph{kit}.
This section shows that the stability provides a new semantic correctness criterion for \( \mathsf{MLL}+\mathsf{Mix} \), which we call the \emph{creed-kit criterion}.
The creed-kit criterion is one of the key steps towards the full definability proof, and we think it should be of independent interest.

\subsection{Creed and Kit}
Just as morphisms from a singleton set play a crucial role in the category \( \Set \), profunctors from \( I \) (or to \( I \)) play an important role in \( \Prof \).
A profunctor \( \Profunctor \colon I \profarrow \Grp \) can be expressed as a sum
\( \Profunctor \cong \coprod_i \InducedProf{\Group_i} \) of (profunctors induced by) subgroups \( \Group_i \le \Grp(a_i,a_i) \).
\emph{Creed} proposed by Taylor~\cite{Taylor1989} and \emph{kit} proposed by Fiore~et~al.~\cite{Fiore2022,Fiore2024} are both designed to appropriately constrain these subgroups \( \Group_i \).

\begin{definition}[Creed~\cite{Taylor1989} and Kit~\cite{Fiore2022,Fiore2024}]
    Let \( \Profunctor \colon I \profarrow \Grp \).
    The \emph{creed of \( \Profunctor \)}, written \( \CreedOf \Profunctor \), is a family of subsets \( (\CreedOf \Profunctor)_a \subseteq \Grp(a,a) \) parameterized by objects \( a \in \Grp \) defined by
    \begin{align*}
        (\CreedOf \Profunctor)_a
        &\quad\defeq\quad
        \{\, \alpha \in \Grp(a,a) \mid \exists x \in \Profunctor(\star, a).\: x \cdot \alpha = x \,\}
        \\
        &\quad=\quad
        \textstyle
        \bigcup_{x \in \Profunctor(a)} \Stabiliser(x).
    \end{align*}
    The \emph{kit of\/ \( \Profunctor \)}, written \( \KitOf \Profunctor \), is a family of subsets \( (\KitOf \Profunctor)_a \subseteq \{ \Group \mid \Group \le \Grp(a,a) \} \) of subgroups of \( \Grp(a,a) \) parameterized by objects \( a \in \Grp \) defined by
    \begin{equation*}
        (\KitOf \Profunctor)_a
        \quad\defeq\quad
        \{\, \Stabiliser(x) \mid x \in \Profunctor(\star, a) \,\}.
    \end{equation*}
    Obviously, \( (\CreedOf \Profunctor)_a = \bigcup (\KitOf \Profunctor)_a \).
    The creed and kit of \( \Profunctorb \colon \Grp \profarrow I \) is defined similarly.
    A family \( C = (C_a)_a \) of subsets \( C_a \subseteq \Grp(a,a) \) is a \emph{creed} if \( C = (\CreedOf \Profunctor) \) for some \( \Profunctor \colon I \profarrow \Grp \).
    Similarly, a \emph{kit} is a family obtained as \( (\KitOf \Profunctor) \).
\qed
\end{definition}
Note that \( \CreedOf \Profunctor \) and \( \KitOf \Profunctor \) are not defined (at least at this moment) for general profunctors \( \Profunctor \colon \Grp \profarrow \Grpb \) (with \( \Grp \neq I \neq \Grpb \)).
The definitions given above differ from the original one, but ours are equivalent to theirs as expected (\cref{prop:creed:equivalence}).
\begin{proposition}\label{prop:creed:equivalence}
    Let \( \Grp \) be a groupoid.
    \begin{enumerate}
        \item A family \( C = (C_a \subseteq \Grp(a,a))_a \) is a creed if and only if
        \begin{enumerate}
\item \( \alpha \in C_a \) implies \( \alpha^k \in C_a \) for every \( k \in \Int \), and
            \item \( \alpha \in C_a \) and \( \beta \in \Grp(a,b) \) implies \( \beta^{-1} ; \alpha ; \beta \in C_b \).
        \end{enumerate}
        \item A family \( K = (K_a \subseteq \{ \Group \mid \Group \le \Grp(a,a) \})_{a} \) is a kit if and only if \( K \) is closed under the conjugation in the sense that, if \( \Group \in K_a \) and \( \beta \in \Grp(a,b) \), then \( \beta^{-1} \cdot \Group \cdot \beta \in K_b \), where \( \beta^{-1} \cdot \Group \cdot \beta = \{ \beta^{-1} ; \alpha ; \beta \mid \alpha \in \Group \} \).
\qed
    \end{enumerate}
\end{proposition}

\begin{proposition}\label{prop:creed:induced-prof}
    Let \( \Grp \) be a groupoid, \( a \in \Grp \) and \( \Group \subgroup \Grp(a,a) \).
    Then \( (\CreedOf{\InducedProf{\Group}})_b = \{\, \beta^{-1} \cdot \alpha \cdot \beta \mid \alpha \in \Group, \beta \in \Grp(a,b) \,\} \) and \( (\KitOf{\InducedProf{\Group}})_b = \{\, \beta^{-1} \cdot \Group \cdot \beta \mid \beta \in \Grp(a,b) \,\} \).
\end{proposition}

\begin{definition}
    For a profunctor \( \Profunctor \colon I \profarrow \Grp \) and a creed \( \creed \) on \( \Grp \), we write \( \Profunctor \models \creed \) to mean \( (\CreedOf \Profunctor)_a \subseteq \creed_a \) for every \( a \in \Grp \).
    Similarly, for a kit \( \kit \), \( \Profunctor \models \kit \) if \( (\KitOf \Profunctor)_a \subseteq \kit_a \) for every \( a \).
\qed
\end{definition}

\begin{proposition}
    Let \( \Profunctor \colon I \profarrow \Grp^{\op} \) and assume \( \Profunctor \cong \coprod_i \InducedProf{\Group_i} \), \( \Group_i \le \Grp(a_i,a_i) \).
    Let \( C \) be a creed and \( K \) be a kit on \( \Grp \).
    \begin{itemize}
        \item \( \Profunctor \models C \) if and only if \( \Group_i \subseteq C_{a_i} \) for every \( i \).
        \item \( \Profunctor \models K \) if and only if \( \Group_i \in K_{a_i} \) for every \( i \).
\qed
    \end{itemize}
\end{proposition}

Now, as we begin to develop the theory using creed and kit, it is worth noting that there are some subtle differences between the approaches of Taylor~\cite{Taylor1989} and Fiore~et~al.~\cite{Fiore2022,Fiore2024}.
This paper follows the approach of Fiore~et~al., mainly due to the fact that the approach of Fiore~et~al.~is based on a well-developed theory of \emph{orthogonality} by Hyland and Schalk~\cite{Hyland2003}.
While we will discuss both creed and kit throughout this paper, it should be noted that some concepts and operations on creed are not always identical to Taylor's original formulation.

\begin{definition}\label{def:creed:jointly-orthogonal}
    Let \( C \) and \( C' \) be creeds on \( \Grp \).
    They are \emph{orthogonal}, written \( C \bot C' \), if and only if \( C_a \cap C'_a \subseteq \{ \ident \} \) for every \( a \).
    For kits \( K \) and \( K' \), their orthogonal is defined by
    \begin{equation*}
        K \bot K' \Longleftrightarrow \forall a. \forall \Group \in K_a. \forall \Group' \in K'_a. \Group \cap \Group' = \{ \ident \}.
    \end{equation*}
    Given a creed \( C \) and kit \( K \) on \( \Grp \), we write \( C^{\bot} \) and \( K^{\bot} \) for the maximum creed and kit such that \( C \bot C^{\bot} \) and \( K \bot K^{\bot} \).
    Concretely, the creed \( C^{\bot} \) and kit \( K^{\bot} \) are given by
    \begin{align*}
        (C^\bot)_a &= \{ \alpha \in \Grp(a,a) \mid \forall k \in \Int. \alpha^k \notin (C_a \setminus \{ \ident \}) \} \\ 
        (K^\bot)_a &= \{ \Group' \le \Grp(a,a) \mid \forall \Group \in K_a. \Group \bot \Group' \}.
    \end{align*}
    The orthogonality extends to profunctors: for \( \Profunctor \colon I \profarrow \Grp \) and \( \Profunctorb \colon \Grp \profarrow I \), we write \( \Profunctor \bot \Profunctorb \) when \( (\CreedOf \Profunctor) \bot (\CreedOf \Profunctorb) \), which is equivalent to \( (\KitOf \Profunctor) \bot (\KitOf \Profunctorb) \)
\qed
\end{definition}
\begin{lemma}\label{lem:creed:orthogonality-for-groups}
    Let \( \Grp \) be a groupoid, \( a \in \Grp \), \( \Group \le \Grp(a,a) \) and \( \Groupb \le \Grp^{\op}(a,a) \).
    Then \( \InducedProf{\Group} \orthogonal \InducedProf{\Groupb} \) if and only if\/ \( \Group' \cap \Groupb' = \{ \ident \} \) for every \( \Group' \Conjugate \Group \) and \( \Groupb' \Conjugate \Groupb \).
\qed
\end{lemma}

Once the orthogonality \( \bot \) is defined, we simply follow Hyland and Schalk~\cite{Hyland2003}.
A creed \( C \) (resp.~kit \( K \)) is \emph{Boolean} if \( C = C^{\bot\bot} \) (resp.~\( K = K^{\bot\bot} \)).
There is a bijective correspondence between Boolean creeds and Boolean kits: \( C \mapsto K \) with \( K_a \defeq \{ \Group \le \Grp(a,a) \mid \Group \subseteq C_a \} \).
\begin{definition}
    The category \( \SProf \) of \emph{stable profunctors} is defined as follows.
    Its object is a pair \( \creed = (\Grp, \creed) \) of a groupoid \( \Grp \) and a creed \( \creed \) on \( \Grp \).
    Its morphism \( (\Grp, \creed) \profarrow (\Grpb, \creedb) \) is a profunctor \( \Profunctor \colon \Grp \profarrow \Grpb \) that satisfies
    \begin{itemize}
        \item \( (\Profunctorb;\Profunctor) \models \creedb \) for every \( \Profunctorb \colon I \profarrow \Grp \) with \( \Profunctorb \models \creed \).
        \item \( (\Profunctor;\Profunctorc) \models \creed^{\bot} \) for every \( \Profunctorc \colon \Grpb \profarrow I \) with \( \Profunctorc \models \creedb^\bot \).
    \end{itemize}
    Morphisms are composed as profunctors.
\qed
\end{definition}

The category of \( \SProf \) is a model of linear logic~\cite{Fiore2024}.
Furthermore, there is the forgetful functor \( \Underlying \colon \SProf \longrightarrow \Prof \) that preserves the structures.
Its action on objects is \( \Underlying(\creed) = \Grp \) when \( \creed = (\Grp, \creed) \).
We describe the structures.

The dual is easy: \( (\Grp, \creed)^{\bot} = (\Grp^{\op}, \creed^{\bot}) \) (note that a creed \( \creed \) on \( \Grp \) can be seen as a creed on \( \Grp^{\op} \)).
The product in \( \SProf \) is actually the biproduct: \( (\Grp, \creed) \oplus (\Grpb, \creedb) = (\Grp \oplus \Grpb, \creed \oplus \creedb) \) where \( \creed \oplus \creedb \) assigns \( \creed_a \) for objects \( a \) in \( \Grp \) and \( \creedb_b \) for objects \( b \) in \( \Grpb \).
The tensor product \( \otimes \) is given by \( (\Grp, \creed) \otimes (\Grpb, \creedb) = (\Grp \times \Grpb, (\creed \otimes \creedb)_{a,b}) \) where \( (\creed \otimes \creedb)_{a,b} = (\creed_a \times \creed_b)^{\bot\bot} \).
Then \( \creed \llpar \creedb \defeq (\creed^\bot \otimes \creedb^\bot)^{\bot} \).

\begin{lemma}\label{lem:creed:creed-sufficient-condition}
    Let \( \creed \) and \( \creedb \) be creeds over \( \Grp \) and \( \Grpb \), respectively.
    Let \( a \in \Grp \), \( \alpha \in \Grp(a,a) \), and
    \begin{gather*}
        \alpha_{++} \in (\creed_a \setminus \{ \ident \}),
        \quad
        \alpha_{+} \in \creed_a,
        \\
        \alpha_{--} \in (\creed^{\bot}_a \setminus \{ \ident \}),
        \quad
        \alpha_{-} \in \creed^{\bot}_a,
    \end{gather*}
    and similarly, \( b \in \Grpb \), \( \beta_{++}, \beta_{+}, \beta, \beta_{-}, \beta_{--} \) as above.
    Then
    \begin{gather*}
        (\alpha_{+}, \beta_{+}) \in (\creed \otimes \creedb)_{(a,b)},
        \quad
        (\alpha_{--}, \beta) \in (\creed \otimes \creedb)^{\bot}_{(a,b)},
        \\
        (\alpha_{++}, \beta) \in (\creed \llpar \creedb)_{(a,b)},
        \quad
        (\alpha_{-}, \beta_{-}) \in (\creed \llpar \creedb)^{\bot}_{(a,b)}.
    \end{gather*}
\qed
\end{lemma}

\subsection{Creed-Kit Criterion}

We characterize the correctness of a proof structure \( \pi \vdash \Gamma \) in terms of creed and kit.
The interpretation of a proof structure \( \pi \) is a profunctor \( \sem{\pi}_{\vec{\Grp}} \colon I \profarrow \sem{\Gamma}_{\vec{\Grp}} \) parameterized by \( \vec{\Grp} = (\Grp_1,\dots,\Grp_N) \).
For each \( \Grp_i \), let \( \creed_i \) be a Boolean creed on \( \Grp_i \) and \( \vec{\creed} = (\creed_1,\dots,\creed_N) \).
Then \( \sem{\Gamma}_{\vec{\creed}} \) is a creed on \( \sem{\Gamma}_{\vec{\Grp}} \), so the stability of \( \sem{\pi}_{\vec{A}} \) with respect to \( \sem{\Gamma}_{\vec{\creed}} \) makes sense.

\begin{definition}[Creed-Kit criterion]
    A cut-free proof structure \( \pi \) satisfies the \emph{creed-kit criterion} if
    \( \sem{\pi}_{\Underlying(\vec{\creed})} \models \sem{\Gamma}_{\vec{\creed}} \) for every assignment \( \vec{\creed} \) of Boolean creeds to propositional variables.
\qed
\end{definition}

The main result of this section is as follows.
\begin{theorem}\label{thm:creed-kit-correctness-without-cut}
    A cut-free proof structure is correct with respect to \( \mathsf{MLL+Mix} \) if and only if it satisfies the creed-kit criterion.
\qed
\end{theorem}

The left-to-right direction is trivial since \( \SProf \) is a model of \( \mathsf{MLL+Mix} \)~\cite{Fiore2024}.
The remainder of this section is devoted to proving the converse direction.

\begin{remark}
    The cut-freeness assumption is essential to \cref{thm:creed-kit-correctness-without-cut}.
    There exists an incorrect proof structure with \( (\mathsf{Cut}) \) of which all sources of its incorrectness are eliminated by the cut elimination procedure.
    An example of such a proof structure (written as an \( \mathsf{MLL}+\mathsf{Mix}+(\otimes{=}\llpar) \) proof) is
    \begin{equation*}
        \pi_0 \defeq\quad
        \infer[(\mathsf{Cut})]{\vdash \PropVar, \PropVar^{\bot}}{
            \infer[(\mathsf{Ex})]{\vdash \PropVar, \PropVar^{\bot}, \PropVar^{\bot} \llpar \PropVar}{
                \infer[(\llpar)]{\vdash \PropVar, \PropVar^{\bot} \llpar \PropVar, \PropVar^{\bot}}{
                    \infer[(\mathsf{Mix})]{\vdash \PropVar, \PropVar^{\bot}, \PropVar, \PropVar^{\bot}}{
                        \infer[(\mathsf{Ax})]{\vdash \PropVar, \PropVar^\bot}{}
                        &
                        \infer[(\mathsf{Ax})]{\vdash \PropVar, \PropVar^\bot}{}
                    }
                }
            }
            &
            \infer[(\otimes{=}\llpar)]{\vdash \PropVar \otimes \PropVar^{\bot}}{
                \infer[(\mathsf{Ax})]{\vdash \PropVar, \PropVar^{\bot}}{}
            }
        }
    \end{equation*}
    to which the cut-elimination procedure gives \( \pi_1 \defeq \: \infer[(\mathsf{Ax})]{\vdash \PropVar, \PropVar^\bot}{} \).
    Since the semantic interpretation is an invariant of the cut-elimination procedure, any semantic criterion cannot distinguish between the correct proof \( \pi_1 \) from the incorrect proof \( \pi_0 \).
\qed
\end{remark}

\subsection{Profunctorial Experiment}\label{sec:experiments}
Before giving the proof of \cref{thm:creed-kit-correctness-without-cut}, we introduce a convenient graphical method for representing the profunctor interpretation of proof nets.
This can be viewed as an extension of Girard's \emph{experiment}~\cite{Girard1987}, which provides a graphical presentation of the relational interpretation of proof nets.

\begin{definition}
    Let \( \Gamma \) be a formula with propositional variables \( \vec{\PropVar} = (\PropVar_1,\dots,\PropVar_n) \), \( \pi \) be a proof structure of \( \Gamma \), and \( \vec{\Grp} = (\Grp_1,\dots,\Grp_n) \) be an assignment of groupoids to propositional variables.
    A \emph{(profunctorial) experiment} is a labelling on its nodes and/or axiom links.

    An \emph{object experiment} is a labelling on its nodes.
    \begin{itemize}
        \item Each occurrence of a propositional variable \( \PropVar_i \) and its negation \( \PropVar_i^{\bot} \) is labelled by an object \( a \in \Grp_i \).
            Different occurrences of the same variable may be assigned different objects.
        \item Each occurrence of \( \otimes \) and \( \llpar \) is labelled by the pair \( (a_1, a_2) \) of objects assigned to its left and right children.
    \end{itemize}
    When \( \Gamma = (\Formula_1, \dots, \Formula_k) \), an object experiment labels an object \( a_i \in \sem{\Formula_i}_{\vec{\Grp}} \) for each \( i \).
    Then \( (a_1,\dots,a_k) \) is an object of \( \sem{\Gamma}_{\vec{\Grp}} \).

    An \emph{element experiment} is a labelling on axiom links.
    \begin{itemize}
        \item Each axiom link \( \PropVar_i^\bot \)---\( \PropVar_i \) is labelled by a morphism \( \alpha \colon a' \longrightarrow a \) in \( \Grp_i \).
    \end{itemize}
    An element experiment determines an object experiment by choosing \( a' \) for \( \PropVar_i^\bot \) and \( a \) for \( \PropVar_i \) when \( \PropVar_i^\bot \)---\( \PropVar_i \) is annotated by \( \alpha \colon a' \longrightarrow a \).
    Let \( \mathit{Exp}(\pi)(a_1,\dots,a_k) \) be the set of all element experiments whose associated object experiment is \( (a_1,\dots,a_k) \).

    Let \( (\alpha_1,\dots,\alpha_k) \colon (a_1,\dots,a_k) \longrightarrow (a'_1,\dots,a'_k) \) be a morphism in \( \sem{\Gamma} \).
    The action of this morphism on \( \mathit{Exp}(\pi)(a_1,\dots,a_k) \) is defined by propagating the morphism upward from the bottom, following the structure of the proof net, and composing whenever it reaches an axiom link.
    This action makes \( \mathit{Exp}(\pi) \) a profunctor \( I \profarrow \sem{\Gamma} \).
\end{definition}

\begin{example}
    An element experiment for the proof structure in \cref{fig:pre:incorrect-proof-structure} is shown in \cref{fig:criterion:element-experiment}.
    The action of \( (((f_2, g_2), g_1), ((h_2, h_1),f_1)) \) to the element experiment in \cref{fig:criterion:element-experiment} is shown in \cref{fig:criterion:experiment-action}.
\qed
\end{example}

\begin{figure}[t]
    \begin{tikzpicture}[yscale=0.6]
	\begin{pgfonlayer}{nodelayer}
		\node [style=none] (00) at (0.5, -0.7) {};
		\node [style=none] (0) at (0.5,0) {$\llpar$};
		\node [style=none] (1) at (-1, 1) {$\otimes$};
		\node [style=none] (2) at (-2, 2) {$\llpar$};
		\node [style=none] (3) at (-2.5, 3) {$\PropVar_1$};
		\node [style=none] (4) at (-1.5, 3) {$\PropVar_2$};
		\node [style=none] (5) at (-0.5, 2) {$\PropVar_2^\bot$};
		\node [style=none] (6) at (2, 1) {$\llpar$};
		\node [style=none] (7) at (1, 2) {$\llpar$};
		\node [style=none] (8) at (0.5, 3) {$\PropVar_2$};
		\node [style=none] (9) at (1.5, 3) {$\PropVar_2^\bot$};
		\node [style=none] (10) at (2.5, 2) {$\PropVar_1^\bot$};
		\node [style=none] (11) at (-2.5, 4.3) {};
		\node [style=none] (12) at (-1.5, 3.5) {};
		\node [style=none] (13) at (-0.5, 3.5) {};
		\node [style=none] (14) at (0.5, 3.5) {};
		\node [style=none] (15) at (1.5, 3.5) {};
		\node [style=none] (16) at (2.5, 4.3) {};
		\node [style=none] (m1) at (0, 4.6) {$\alpha \colon a_1 \to a_2$};
		\node [style=none] (m2) at (-1, 3.8) {$\beta \colon b_1 \to b_2$};
		\node [style=none] (m3) at (1, 3.8) {$\gamma \colon c_1 \to c_2$};
	\end{pgfonlayer}
	\begin{pgfonlayer}{edgelayer}
		\draw (0) to (00.center);
		\draw (0) to (1);
		\draw (1) to (2);
		\draw (1) to (5);
		\draw (2) to (3);
		\draw (2) to (4);
		\draw (0) to (6);
		\draw (6) to (7);
		\draw (6) to (10);
		\draw (7) to (8);
		\draw (7) to (9);
		\draw (3) to (11.center);
		\draw (4) to (12.center);
		\draw (5) to (13.center);
		\draw (8) to (14.center);
		\draw (9) to (15.center);
		\draw (10) to (16.center);
		\draw (11.center) to (16.center);
		\draw (12.center) to (13.center);
		\draw (14.center) to (15.center);
	\end{pgfonlayer}
\end{tikzpicture}
     \caption{An element experiment.  Here \( \alpha \colon a_1 \to a_2 \) is in \( \Grp_1 \), and \( \beta \colon b_1 \to b_2 \) and \( \gamma \colon c_1 \to c_2 \) are morphisms in \( \Grp_2 \).  This describes an element in \( \mathit{Exp}(\pi)(((a_2, b_2), b_1),((c_2, c_1), a_1)) \).}
    \label{fig:criterion:element-experiment}
\end{figure}

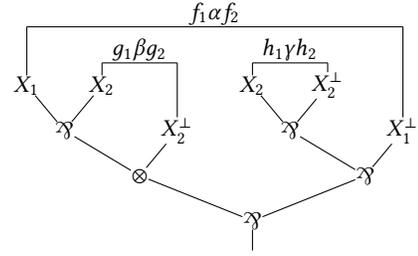
\begin{figure}[t]
    \begin{tikzpicture}[yscale=0.6]
	\begin{pgfonlayer}{nodelayer}
		\node [style=none] (00) at (0.5, -0.7) {};
		\node [style=none] (0) at (0.5,0) {$\llpar$};
		\node [style=none] (1) at (-1, 1) {$\otimes$};
		\node [style=none] (2) at (-2, 2) {$\llpar$};
		\node [style=none] (3) at (-2.5, 3) {$\PropVar_1$};
		\node [style=none] (4) at (-1.5, 3) {$\PropVar_2$};
		\node [style=none] (5) at (-0.5, 2) {$\PropVar_2^\bot$};
		\node [style=none] (6) at (2, 1) {$\llpar$};
		\node [style=none] (7) at (1, 2) {$\llpar$};
		\node [style=none] (8) at (0.5, 3) {$\PropVar_2$};
		\node [style=none] (9) at (1.5, 3) {$\PropVar_2^\bot$};
		\node [style=none] (10) at (2.5, 2) {$\PropVar_1^\bot$};
		\node [style=none] (11) at (-2.5, 4.3) {};
		\node [style=none] (12) at (-1.5, 3.5) {};
		\node [style=none] (13) at (-0.5, 3.5) {};
		\node [style=none] (14) at (0.5, 3.5) {};
		\node [style=none] (15) at (1.5, 3.5) {};
		\node [style=none] (16) at (2.5, 4.3) {};
		\node [style=none] (m1) at (0, 4.6) {$f_1 \alpha f_2$};
		\node [style=none] (m2) at (-1, 3.8) {$g_1 \beta g_2$};
		\node [style=none] (m3) at (1, 3.8) {$h_1 \gamma h_2$};
	\end{pgfonlayer}
	\begin{pgfonlayer}{edgelayer}
		\draw (00) to (0);
		\draw (0) to (1);
		\draw (1) to (2);
		\draw (1) to (5);
		\draw (2) to (3);
		\draw (2) to (4);
		\draw (0) to (6);
		\draw (6) to (7);
		\draw (6) to (10);
		\draw (7) to (8);
		\draw (7) to (9);
		\draw (3) to (11.center);
		\draw (4) to (12.center);
		\draw (5) to (13.center);
		\draw (8) to (14.center);
		\draw (9) to (15.center);
		\draw (10) to (16.center);
		\draw (11.center) to (16.center);
		\draw (12.center) to (13.center);
		\draw (14.center) to (15.center);
	\end{pgfonlayer}
\end{tikzpicture}
     \caption{The result of the action of \( (((f_2, g_2), g_1),((h_2, h_1), f_1)) \) on the element experiment in \cref{fig:criterion:element-experiment}.}
    \label{fig:criterion:experiment-action}
\end{figure}

\begin{lemma}
    \( \sem{\pi} \cong \mathit{Exp}(\pi) \).
\end{lemma}
\begin{proof}
    By easy induction on the size of proof structures.
\end{proof}

\subsection{Proof of \cref{thm:creed-kit-correctness-without-cut}}
We prove the right-to-left direction by contraposition.
Assume an incorrect proof structure \( \pi \).

We can assume without loss of generality that \( \pi \vdash \Formula \) for some formula \( \Formula \); otherwise, add \( \llpar \)-nodes as many as needed to make a sequent into one formula.
By \cref{thm:weak-danos-regnier}, there exists a Danos-Regnier switching of \( \pi \) such that the associated graph has a cycle.
Fix a Danos-Regnier switching with a cycle in the associated graph and choose a cycle.
Let \( N \) be the node such that (1) it is involved in the cycle and that (2) the ancestors of \( N \) are not involved in the cycle.
The node \( N \) must be labelled by \( \otimes \).
We can assume without loss of generality that \( N \) is reachable from the root in the associated graph; if it is not reachable, change the switching in the unique path from \( N \) to the root in the syntax tree.
\begin{example}
    For the incorrect proof structure in \cref{fig:pre:incorrect-proof-structure}, the node \( N \) is the unique \( \otimes \)-node.
\qed
\end{example}

We assign directions to some of the edges using the following rules:
\begin{itemize}
    \item The edges involved in the cycle from \( N \) to \( N \) are oriented so that the traversal from \(N\) goes to the upper-left of \( N \) and then returns from the upper-right.
    \item The edges in the path in the syntax tree from \( N \) to the root are oriented toward the root.
\end{itemize}
\begin{example}
    For the incorrect proof structure in \cref{fig:pre:incorrect-proof-structure}, these rules results in the partially directed graph in \cref{fig:criterion:partial-orientation}.
\qed
\end{example}
\begin{figure}[t]
    \begin{tikzpicture}[yscale=0.6]
	\begin{pgfonlayer}{nodelayer}
		\node [style=none] (00) at (0.5, -0.7) {};
		\node [style=none] (0) at (0.5,0) {$\llpar$};
		\node [style=none] (1) at (-1, 1) {$\otimes$};
		\node [style=none] (2) at (-2, 2) {$\llpar$};
		\node [style=none] (3) at (-2.5, 3) {$\PropVar_1$};
		\node [style=none] (4) at (-1.5, 3) {$\PropVar_2$};
		\node [style=none] (5) at (-0.5, 2) {$\PropVar_2^\bot$};
		\node [style=none] (6) at (2, 1) {$\llpar$};
		\node [style=none] (7) at (1, 2) {$\llpar$};
		\node [style=none] (8) at (0.5, 3) {$\PropVar_2$};
		\node [style=none] (9) at (1.5, 3) {$\PropVar_2^\bot$};
		\node [style=none] (10) at (2.5, 2) {$\PropVar_1^\bot$};
		\node [style=none] (11) at (-2.5, 4) {};
		\node [style=none] (12) at (-1.5, 3.5) {};
		\node [style=none] (13) at (-0.5, 3.5) {};
		\node [style=none] (14) at (0.5, 3.5) {};
		\node [style=none] (15) at (1.5, 3.5) {};
		\node [style=none] (16) at (2.5, 4) {};
	\end{pgfonlayer}
	\begin{pgfonlayer}{edgelayer}
		\draw[<-] (00) to (0);
		\draw[<-] (0) to (1);
		\draw[->] (1) to (2);
		\draw[<-] (1) to (5);
		\draw[dashed] (2) to (3);
		\draw[->] (2) to (4);
		\draw[dashed] (0) to (6);
		\draw[dashed] (6) to (7);
		\draw[dashed] (6) to (10);
		\draw[dashed] (7) to (8);
		\draw[dashed] (7) to (9);
		\draw[dashed] (3) to (11.center);
		\draw (4) to (12.center);
		\draw[<-] (5) to (13.center);
		\draw[dashed] (8) to (14.center);
		\draw[dashed] (9) to (15.center);
		\draw[dashed] (10) to (16.center);
		\draw[dashed] (11.center) to (16.center);
		\draw (12.center) to (13.center);
		\draw[dashed] (14.center) to (15.center);
	\end{pgfonlayer}
\end{tikzpicture}
     \caption{The result of the orientation.  For readability, unoriented edges are drawn as dashed lines.}
    \label{fig:criterion:partial-orientation}
\end{figure}
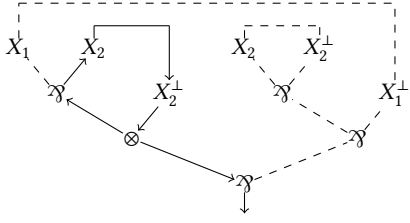

\begin{claim*}
    The edges to the start node \( N \) must be oriented as
    
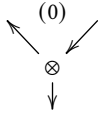
\begin{gather*}
        \xymatrix@C=3mm@R=3mm{
            &  (0) & \ar[ld] \\
            &
            {\otimes} \ar[ul] \ar[d] &
            \\
            & &
        }
    \end{gather*}
\end{claim*}
\begin{proof}[Proof of the claim]\renewcommand{\qedsymbol}{$\blacksquare$}
    Obvious from the definition.
\end{proof}

\begin{claim*}
    The orientation of a \( \otimes \)-node other than \( N \) is one of the following 
    (modulo the swapping of left and right children):
    
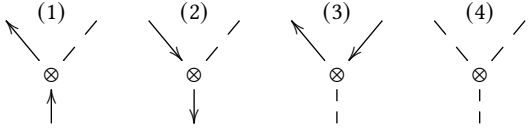
\begin{gather*}
        \xymatrix@C=3mm@R=4mm{
            & \tiny (1) & \\
            &
            {\otimes} \ar[ul] \ar@{--}[ur] &
            \\
            & \ar[u] &
        }
        \quad
        \xymatrix@C=3mm@R=4mm{
            \ar[dr] & \tiny (2) & \\
            &
            {\otimes} \ar@{--}[ur] \ar[d] &
            \\
            & &
        }
        \quad
        \xymatrix@C=3mm@R=4mm{
            & \tiny (3) & \ar[dl] \\
            &
            {\otimes} \ar[ul] \ar@{--}[d] &
            \\
            & &
        }
        \quad
        \xymatrix@C=3mm@R=4mm{
            \ar@{--}[dr] & \tiny (4) & \\
            &
            {\otimes} \ar@{--}[ur] &
            \\
            & \ar@{--}[u] &
        }
    \end{gather*}
\end{claim*}
\begin{proof}[Proof of the claim]\renewcommand{\qedsymbol}{$\blacksquare$}
    Because \(\otimes\)-nodes other than \( N \) are either not visited or are simply passed through.
\end{proof}

\begin{claim*}
    The orientation of a \( \llpar \)-node is one of the following (modulo the swapping of left and right children):
    
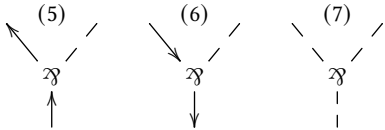
\begin{gather*}
        \xymatrix@C=3mm@R=4mm{
            & \tiny (5) & \\
            &
            {\llpar} \ar[ul] \ar@{--}[ur] &
            \\
            & \ar[u] &
        }
        \quad
        \xymatrix@C=3mm@R=4mm{
            \ar[dr] & \tiny (6) & \\
            &
            {\llpar} \ar[d] \ar@{--}[ur] &
            \\
            & &
        }
        \quad
        \xymatrix@C=3mm@R=4mm{
            & \tiny (7) & \\
            &
            {\llpar} \ar@{--}[ul] \ar@{--}[d] \ar@{--}[ur] &
            \\
            & &
        }
    \end{gather*}
\end{claim*}
\begin{proof}[Proof of the claim]\renewcommand{\qedsymbol}{$\blacksquare$}
    An unvisited node is oriented as \( (7) \).
    A node in the path from \( N \) to the root is oriented as \( (6) \).
    Otherwise, a node is in the cycle.
    The claim says that
    
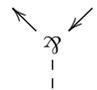
\begin{equation*}
        \xymatrix@C=3mm@R=3mm{
            & & \ar[dl] \\
            &
            {\llpar} \ar[ul] \ar@{--}[d] &
            \\
            & &
        }
    \end{equation*}
    is prohibited.
    This is because all edges in the cycle are those selected by a Danos-Regnier switch.
\end{proof}
    
Let \( \Grp \) be a single-object groupoid and \( \creed \) be a creed on \( \Grp \) such that both \( \creed_a \) and \( \creed^{\bot}_a \) are non-trivial, where \( a \in \Grp \) is the unique object.
For example, \( \Grp(a,a) = \Int_2 \times \Int_3 \), \( \creed_a = \Int_2 \times \{ 0 \} \) and \( \creed^\bot_a = \{ 0 \} \times \Int_3 \).
We choose \( \alpha_+ \in (\creed_a \setminus \creed^\bot_a) \) and \( \alpha_- \in (\creed^{\bot}_a \setminus \creed_a) \).
Let \( \vec{\creed} = (\creed, \dots, \creed) \) and \( \vec{\Grp} = \Underlying(\vec{\creed}) \).

Consider an element experiment \( x \) for \( \pi \) in which all labels are \( \ident_a \).
For each occurrence of a literal, we assign an endo-morphism in \( \Grp(a,a) \) by the following rules:
\begin{itemize}
    \item For an axiom link oriented \( \PropVar \!\to\! \PropVar^{\bot} \), we assign \( \alpha_+ \) for \( \PropVar \) and \( \alpha_+^{-1} \) for \( \PropVar^{\bot} \).
    \item For an axiom link oriented \( \PropVar \!\leftarrow\! \PropVar^{\bot} \), we assign \( \alpha_- \) for \( \PropVar \) and \( \alpha_-^{-1} \) for \( \PropVar^{\bot} \).
    \item For an axiom link with no orientation \( \PropVar \)---\( \PropVar^{\bot} \), we assign \( \ident_a \) for both \( \PropVar \) and \( \PropVar^{\bot} \).
\end{itemize}
The assignment determines an endo-morphism \( \alpha \) in \( \sem{A} \).

\begin{example}
    Consider the partially directed graph in \cref{fig:criterion:partial-orientation}.
    The assignment is shown in \cref{fig:criterion:stabilizer-assignment}.
\qed
\end{example}
\begin{figure}[t]
    \begin{tikzpicture}[yscale=0.6]
	\begin{pgfonlayer}{nodelayer}
		\node [style=none] (00) at (0.5, -0.7) {};
		\node [style=none] (0) at (0.5,0) {$\llpar$};
		\node [style=none] (1) at (-1, 1) {$\otimes$};
		\node [style=none] (2) at (-2, 2) {$\llpar$};
		\node [style=none] (3) at (-2.5, 3) {$\PropVar_1$};
		\node [style=none] (4) at (-1.5, 3) {$\PropVar_2$};
		\node [style=none] (5) at (-0.5, 2) {$\PropVar_2^\bot$};
		\node [style=none] (6) at (2, 1) {$\llpar$};
		\node [style=none] (7) at (1, 2) {$\llpar$};
		\node [style=none] (8) at (0.5, 3) {$\PropVar_2$};
		\node [style=none] (9) at (1.5, 3) {$\PropVar_2^\bot$};
		\node [style=none] (10) at (2.5, 2) {$\PropVar_1^\bot$};
		\node [style=none] (11) at (-2.5, 4) {};
		\node [style=none] (12) at (-1.5, 3.5) {};
		\node [style=none] (13) at (-0.5, 3.5) {};
		\node [style=none] (14) at (0.5, 3.5) {};
		\node [style=none] (15) at (1.5, 3.5) {};
		\node [style=none] (16) at (2.5, 4) {};
		\node [style=none] (3s) at (-2.8, 3.4) {$\ident$};
		\node [style=none] (4s) at (-1.8, 3.4) {$\alpha_{+}$};
		\node [style=none] (5s) at (-0.8, 2.5) {$\alpha_{+}^{-1}$};
		\node [style=none] (8s) at (0.1, 3.4) {$\ident$};
		\node [style=none] (9s) at (1.8, 3.4) {$\ident$};
		\node [style=none] (10s) at (2.2, 2.4) {$\ident$};
	\end{pgfonlayer}
	\begin{pgfonlayer}{edgelayer}
		\draw[<-] (00) to (0);
		\draw[<-] (0) to (1);
		\draw[->] (1) to (2);
		\draw[<-] (1) to (5);
		\draw[dashed] (2) to (3);
		\draw[->] (2) to (4);
		\draw[dashed] (0) to (6);
		\draw[dashed] (6) to (7);
		\draw[dashed] (6) to (10);
		\draw[dashed] (7) to (8);
		\draw[dashed] (7) to (9);
		\draw[dashed] (3) to (11.center);
		\draw (4) to (12.center);
		\draw[<-] (5) to (13.center);
		\draw[dashed] (8) to (14.center);
		\draw[dashed] (9) to (15.center);
		\draw[dashed] (10) to (16.center);
		\draw[dashed] (11.center) to (16.center);
		\draw (12.center) to (13.center);
		\draw[dashed] (14.center) to (15.center);
	\end{pgfonlayer}
\end{tikzpicture}
     \caption{The assignment of morphisms for \cref{fig:criterion:stabilizer-assignment}.  The assignment determines the endo-morphism \( (((\ident, \alpha_-), \alpha_{-}^{-1}), ((\ident, \ident), \ident)) \) on \( (((a,a),a),((a,a),a)) \).  This is a stabilizer of the element experiment in which all labels on axiom links are identities.}
    \label{fig:criterion:stabilizer-assignment}
\end{figure}
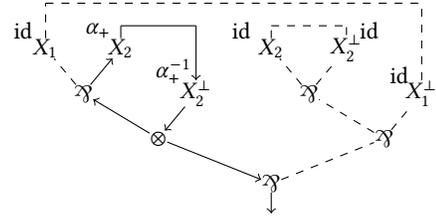

The assignment determines an endo-morphism on each node.
For example, the left-most \( \llpar \)-node in \cref{fig:criterion:stabilizer-assignment} is associated with \( (\ident, \alpha_+) \colon (a,a) \longrightarrow (a,a) \), which is a morphism in \( \sem{\PropVar_1 \llpar \PropVar_2}_{\vec{\creed}} \).
In particular, it determines an endo-morphism in the interpretation of the root formula \( \sem{A}_{\vec{\creed}} \).
By construction, this endo-morphism is a stabilizer of the element experiment \( x \), in which all axiom links are labelled by the identity.

It suffices to prove that this endo-morphism does not belong to the creed for \( \sem{A}_{\vec{\creed}} \).
We prove the following claim by induction on the structure of the formula.
We define the orientation of a node as the orientation of the edge extending downward from it.
For example, in \cref{fig:criterion:partial-orientation}, the left-most \( \llpar \)-node is upward, whereas the right-most \( \llpar \)-node is unoriented.
\begin{claim*}
    Let \( n \) be a node corresponding to a formula \( B \) and \( \creedb = \sem{B}_{\vec{\creed}} \).
    \begin{itemize}
        \item If \( n \) is directed downward, the endo-morphism assigned to \( n \) belongs to \( \creedb^\bot(b)\setminus\{\ident\} \) (where \( b \) is the unique object).
        \item If \( n \) is not oriented, the endo-morphism assigned to \( n \) belongs to \( \creedb^\bot(b) \) (where \( b \) is the unique object).
    \end{itemize}
\end{claim*}
\begin{proof}[Proof of the claim]\renewcommand{\qedsymbol}{$\blacksquare$}
    By easy induction (starting from the leaves of the syntax tree).
    We need to check the cases \( (0) \), \( (2) \), \( (3) \), \( (4) \), \( (6) \) and \( (7) \).
    It suffices to directly apply \cref{lem:creed:creed-sufficient-condition}.
\end{proof}

Since the edge below the root node is oriented downward, \( \gamma \) belongs to \( (\sem{\Formula}_{\vec{\creed}})^\bot_\star \) (where \( \star \) is the unique object of \( \sem{\Formula}_{\vec{\creed}} \))
and \( \gamma \neq \ident \).
So \( \gamma \notin (\sem{\Formula}_{\vec{\creed}})_\star \).
However, \( \gamma \in \Stabiliser(x) \), so \( \sem{\pi}_{\vec{\Grp}} \) violates the creed.

\begin{remark}
    We discuss the relationship between our proof and the proof of Retor{\'e}~\cite{Retore1997}.
    As already mentioned, the two proofs are the same at a high level.
    Our our proof assigns an endomorphism to each edge of a proof structure and check whether it is a morphism allowed by the creed.
    Retor{\'e}'s proof considers two experiments of the same proof structure and discusses whether the pair of elements assigned to each edge is coherent.
    Roughly speaking, the two proofs correspond by reading ``creed'' as ``coherence.''
    
    However, this correspondence is only superficial at present.
    We do not know a way to construct a groupoid equipped with a creed from a coherence space, nor a construction in the opposite direction.
    We beleave that understanding the mathematical background of the above-mentioned correspondence is an important topic for profunctorial models.
\end{remark}
 \section{Totality for Profunctors}\label{sec:totality}
Although the concepts of creeds and kits provide a correctness criterion,
these concepts are insufficient to characterize the profunctors definable by correct \( \mathsf{MLL+Mix} \) proofs.
For example, let \( \pi \) and \( \pi' \) be correct proof structures of the same sequent, say \( \Gamma \), and consider the profunctor \( \sem{\pi} + \sem{\pi'}\).
This is not definable but passes the creed-kit criterion: every \( x \in (\sem{\pi} + \sem{\pi'})(a) \) is either \( x \in \sem{\pi}(a) \) or \( x \in \sem{\pi'}(a) \), so the stabilizer \( \Stabiliser(x) \) is admitted by the creed \( \sem{\Gamma} \).

A property of definable profunctors that creeds and kits do not capture is a certain kind of ``orbit uniqueness.''
To describe this property, we give a profunctor version of totality spaces~\cite{Loader1994}.

\subsection{Preliminaries: Totality Spaces}
Let us briefly review Loader's \emph{totality space}~\cite{Loader1994}, which inspires the construction of our profunctorial model. 
Let \( A \) be a set, and \( U \subseteq \Rel(I, A) \) and \( X \subseteq \Rel(A, I) \).
We identify \( \Rel(I,A) \) and \( \Rel(A, I) \) with the powerset \(\Powerset(A)\) of \( A \).
We define \( U \Bot X \) as follows:
\begin{equation*}
    U \Bot X
    \quad\defiff\quad
    \forall u \in U. \forall x \in X. \#(u \cap x) = 1.
\end{equation*}
We write \( X^\Bot \defeq \{ u \subseteq A \mid \forall x \in X. \#(x \cap u) = 1 \} \).

The category \( \Tot \) of \emph{totality spaces}~\cite{Loader1994}
is given by:
\begin{itemize}
    \item An object is a triple \( (A, U, X) \), where \( A \) is a set and \( U,X \subseteq \Powerset(A) \) such that \( X = U^{\Bot} \), \( U = X^{\Bot} \) and \( \bigcup U = A = \bigcup X \).
    \item A morphism \( R \colon (A, U, X) \longrightarrow (B, V, Y) \) is a relation \( R \subseteq X \times Y \) such that \( \forall u \in U. (u;R) \in V \) and \( \forall y \in Y. (R;y) \in X \). 
\end{itemize}
The category \( \Tot \) is a model of \( \mathsf{MLL+Mix} \).
Loader~\cite{Loader1994} gave a fully definable model of \( \mathsf{MLL+Mix} \) based on \( \Tot \).

An interesting property of a totality space is that every element in \( U \) is maximal with respect to \( \subseteq \).
\begin{lemma}[Lemma~1 in~\cite{Loader1994}]\label{lem:total:loader-maximality}
    Let \( (A, \TotalElement(A), \CototalElement(A)) \) be a totality space.
    If \( x,y \in \TotalElement(A) \) and \( x \subseteq y \), then \( x = y \).
\qed
\end{lemma}
A similar property holds in our model as well, and it plays a crucial role in our proof of full definability.

\newcommand{\TotalOrthogonal}{\Bot}
\subsection{Totality for Profunctors}
In this subsection, we present a profunctorial model corresponding to the totality space model \( \Tot \) and study its basic properties.
To this end, it suffices to equip \( \ProfCat \) with an appropriate orthogonality.
Our orthogonality is obtained by combining the orthogonality \( \Bot \) for the totality space model with the stable orthogonality \( \bot \) for profunctors.
\begin{definition}
    The \emph{stably total orthogonality} \( \TotalOrthogonal \) on \( \ProfCat \) is defined as follows:
    for \( \Profunctor \in \ProfCat(I, \Grp) \) and \( \Profunctorb \in \ProfCat(\Grp, I) \),
    \begin{align*}
\Profunctor \TotalOrthogonal \Profunctorb &\quad\defiff\quad \Profunctor \orthogonal \Profunctorb \mbox{ and } (\Profunctor ; \Profunctorb) \cong \ident_I,
    \end{align*}
    where \( \orthogonal \) is the stable orthogonality in \cref{def:creed:jointly-orthogonal}.
\qed
\end{definition}
\begin{remark}
    The stably total orthogonality is the combination of the orthogonality \( \bot \), which captures the stability, and \( \Profunctorb \circ \Profunctor \cong \ident_I \).
    The former is already incorporated in \( \SProf \), so the totality orthogonality can also be defined as an orthogonality on \( \SProf \) given by
    \begin{equation*}
        \Profunctor \TotalOrthogonal_{\SProf} \Profunctorb \quad\defiff\quad (\Profunctor ; \Profunctorb) \cong \ident_I,     
    \end{equation*}
    for \( \Profunctor \in \SProf(I, \Grp) \) and \( \Profunctorb \in \SProf(\Grp, I) \).
    This is a focused orthogonality, so it is well-behaved (cf.~\cite{Hyland2003}).
\qed
\end{remark}

It is easy to observe that stably total orthogonality corresponds, through the functor \( \Truncation{-} \colon \Prof \longrightarrow \Rel \), to the relational version of the total orthogonality (i.e., \( \Profunctor \TotalOrthogonal \Profunctorb \) implies \( \Truncation{\Profunctor} \Bot \Truncation{\Profunctorb} \)).
Hence, stably total orthogonality may be regarded as a refinement of the total orthogonality on relations that additionally captures the behaviour on morphisms.

To investigate how stably total orthogonality behaves on morphisms, let us first consider the simplest case, namely that of a single-object groupoid (i.e., a group).
Interestingly, the stably total orthogonality coincides with being a strict factorization system.
\begin{proposition}\label{prop:totality:strict-factorization-system}
    Let \( \Grp \) be a groupoid, \( a \in \Grp \), \( \Group \le \Grp(a,a) \) and \( \Groupb \le \Grp^{\op}(a,a) \).
    Then \( \InducedProf{\Group} \TotalOrthogonal \InducedProf{\Groupb} \) if and only if \( (\Group, \Groupb) \) forms a strict factorization system, i.e.~every \( \alpha \in \Grp(a,a) \) factors uniquely as \( \alpha = gh \) with \( g \in \Group \) and \( h \in \Groupb \).
\end{proposition}
\begin{proof}
    Assume that \( \InducedProf{\Group} \TotalOrthogonal \InducedProf{\Groupb} \).
    Since \( \InducedProf{\Group} ; \InducedProf{\Groupb} \cong \ident_I \), we have \( (\Group \cdot \alpha_1, \alpha_2 \cdot \Groupb) \sim (\Group, \Groupb) \) for every \( \alpha_1, \alpha_2 \in \Grp(a,a) \).
    In particular, \( (\Group, \alpha \cdot \Groupb) \sim (\Group, \Groupb) \), so there exists \( \alpha_0 \in \Grp(a,a) \) such that \( (\Group \cdot \alpha_0^{-1}, \alpha_0 \alpha \cdot \Groupb) = (\Group, \Groupb) \), i.e.~\( \alpha_0^{-1} \in \Group \) and \( \alpha_0 \alpha \in \Groupb \).
    So we can take \( g = \alpha_0^{-1} \) and \( h = (\alpha_0 \alpha) \).
This decomposition is unique.
    Assume \( g_1 h_1 = g_2 h_2 \) for some \( g_1, g_2 \in \Group\) and \( h_1, h_2 \in \Groupb \).
    Then \( g_2^{-1} g_1 = h_2 h_1^{-1} \in (\Group \cap \Groupb) \).
Since \( \InducedProf{\Group} \orthogonal \InducedProf{\Groupb} \) implies \( (\Group \cap \Groupb) = \{ \ident \} \) by Lemma~\ref{lem:creed:orthogonality-for-groups},
    we have \( g_2^{-1} g_1 = h_2 h_1^{-1} = \ident \), so \( g_1 = g_2 \) and \( h_1 = h_2 \).

    Conversely, assume that \( \Group \) and \( \Groupb \) form a strict factorization system.
    We first show that \( \InducedProf{\Group} \orthogonal \InducedProf{\Groupb} \) by using Lemma~\ref{lem:creed:orthogonality-for-groups}.
    Suppose \( \Group' \Conjugate \Group \) and \( \Groupb' \Conjugate \Groupb \), and let \( \alpha \in \Group' \cap \Groupb' \).
    Then \( \alpha = (\beta^{-1};g;\beta) = (\gamma;h;\gamma^{-1}) \) for some \( \beta, \gamma \in \Grp(a,a) \), \( g \in \Group \) and \( h \in \Groupb \).
    Let \( \beta' = (\beta;\gamma) \).
    Then \( (g;\beta) = (\beta;h) \).
    We factorize \( \beta \): \( \beta = (g';h') \).
    Then \( ((g;g');h') = (g';(h';h)) \).
    Since \( (\Group,\Groupb) \) is a strict factorization system, we have \( (g;g') = g' \) and \( h' = (h';h) \).
    Hence, \( g = h = \ident \), which implies \( \alpha = \ident \).

    To prove \( \InducedProf{\Group} ; \InducedProf{\Groupb} \cong \ident_I \), it suffices to show that \( (\Group \cdot \alpha_1, \alpha_2 \cdot \Groupb) \sim (\Group, \Groupb) \) for every \( \alpha_1, \alpha_2 \in \Grp(a,a) \).
    Let \( \alpha_1 \alpha_2 = g h \) be the decomposition with \( g \in \Group \) and \( h \in \Groupb \).
    Then \( (\Group \cdot \alpha_1, \alpha_2 \cdot \Groupb) \sim (\Group, \alpha_1 \alpha_2 \cdot \Groupb) = (\Group, gh \cdot \Groupb) \sim (\Group \cdot g, h \cdot \Groupb) = (\Group, \Groupb) \).
\end{proof}

\begin{remark}
    This is not only a key technical ingredient in our proof, but also has a conceptual significance.
    While strict factorization systems are a familiar theme in profunctor models~\cite{Tsukada2018,Olimpieri2021,Clairambault2023,Clairambault2024}, our route to them is quite different, suggesting their apparent fundamental importance from a new perspective.
\qed
\end{remark}

For \( \Groupb, \Groupc \le \Group \), we write \( \Groupb \Bot \Groupc \) if \( (\Groupb, \Groupc) \) is a strict factorization system.
This orthogonality satisfies a maximality property that is close to, but stronger than, \cref{lem:total:loader-maximality}.
\begin{lemma}\label{lem:totality:basic-property-of-total-orthogonality-on-groups}
If \( \Groupb \Bot \Groupc \), \( \Groupb' \Bot \Groupc \) and \( \Groupb \subseteq \Groupb' \), then \( \Groupb = \Groupb' \).
\end{lemma}
\begin{proof}
Let \( \Groupb, \Groupb', \Groupc \le \Group \) and assume \( \Groupb \Bot \Groupc \), \( \Groupb' \Bot \Groupc \) and \( \Groupb \subseteq \Groupb' \).
    Take any \( h' \in \Groupb' \).
    We factorize \( h' \) as \( h' = hk \), \( h \in \Groupb \) and \( k \in \Groupc \).
    Then \( k = h^{-1} hk = h^{-1} h' \in \Groupb' \), so \( k = \ident \) since \( \Groupb' \orthogonal \Groupc \).
    This implies \( h' = h \in \Groupb \).
\end{proof}

For \( U \subseteq \Prof(I, \Grp) \) and \( X \subseteq \Prof(\Grp, I) \),
\begin{align*}
    U^{\TotalOrthogonal} &\defeq \{ x \in \Prof(\Grp, I) \mid \forall u \in U.\: u \TotalOrthogonal x \} \\
    {}^{\TotalOrthogonal}X &\defeq \{ u \in \Prof(I, \Grp) \mid \forall x \in X.\: u \TotalOrthogonal x \}.
\end{align*}
Let \( T_{\TotalOrthogonal}(\ProfCat) \) be the \emph{tight orthogonality category} in the sense of Hyland and Schalk~\cite{Hyland2003}.\footnote{One reason for treating \( \ProfCat \) as a \( 1 \)-category in this paper (by passing to equivalence classes of morphisms) is that this allows us to directly apply the results of Hyland and Schalk~\cite{Hyland2003}.
    Ideally, one would instead work with \( \ProfCat \) as a bicategory, by developing a bicategorical version of the Hyland-Schalk construction (as mentioned in \cite[Section~7.3]{Fiore2024}).
    However, carrying this out would require a substantial amount of additional work that is largely orthogonal to the main technical contributions of the present paper.}
The category \( \TotProf \) of \emph{stably total profunctors} is defined as a full subcategory of \( T_{\TotalOrthogonal}(\ProfCat) \).
\begin{definition}
    The \emph{tight orthogonality category} \( T_{\TotalOrthogonal}(\ProfCat) \) is defined by the following data.
    \begin{itemize}
        \item Its object is a triple \( (\Grp, U, X) \) where \( U \subseteq \Prof(I, \Grp) \) and \( V \subseteq \Prof(\Grp, I) \) such that \( U^{\TotalOrthogonal} = V \) and \( {}^{\TotalOrthogonal}V = U \).
        \item Its morphism from \( (\Grp, U, X) \) to \( (\Grpb, V, Y) \) is a profunctor \( \Profunctor \colon \Grp \profarrow \Grpb \) such that \( (u;\Profunctor) \in V \) for every \( u \in U \) and \( (\Profunctor; y) \in X \) for every \( y \in Y \).
    \end{itemize}
    The identity and composition are inherited from \( \ProfCat \).
    The category \( \TotProf \) of \emph{stably total profunctors} is defined as a full subcategory of \( T_{\TotalOrthogonal}(\ProfCat) \) in which an object \( (\Grp, U, X) \) satisfies \( \forall a \in \Grp. (\exists u \in U. u(a) \neq \emptyset) \wedge (\exists x \in X. x(a) \neq \emptyset) \).
    An object \( \creed \) of \( \TotProf \) is called a \emph{totality groupoid}.
    For a totality groupoid \( \creed = (\Grp, U, X) \), we write \( \TotalElement(\creed) \) for its \( U \) component and \( \CototalElement(\creed) \) for its \( X \) component.
\qed
\end{definition}

\begin{remark}
    The difference between \( \TotProf \) and the tight orthogonality category \( T_{\TotalOrthogonal}(\ProfCat) \) is not new to the profunctorial setting.
    In the relational setting, the difference has already been observed in \cite[Example~66(2)]{Hyland2003}.
\qed
\end{remark}

There is the forgetful functor \( \Underlying \colon \TotProf \longrightarrow \Prof \), which strictly preserves the \( * \)-autonomous structure.

Recall the functor \( \Truncation{-} \colon \ProfCat \longrightarrow \Rel \).
This functor extends to a functor \( \Truncation{-} \colon \TotProf \longrightarrow \Tot \).
The point is that, for \( U \subseteq \Prof(I, \Grp) \), the direct application of \( \Truncation{-} \) yields \( \Truncation{U} = \{ \Truncation{u} \mid u \in U \} \subseteq \Rel(I, \Truncation{\Grp}) \).
\begin{proposition}
    \( \Truncation{-} \colon \TotProf \longrightarrow \Tot \) is a functor that strongly preserves the \( \mathsf{MLL} \) structures.
\qed
\end{proposition}

Every stably total profunctor enjoys the single orbit property.
In particular, the profunctor \( \sem{\pi} + \sem{\pi'} \) mentioned at the beginning of this section is not total.
\begin{lemma}[Single orbit property]\label{lem:total:single-orbit-property}
    Let \( \TGrp \) and \( \TGrpb \) be totality groupoids, \( \Profunctor \colon \TGrp \profarrow \TGrpb \) be a total profunctor,
    \( x \in \Profunctor(a,b) \) and \( x' \in \Profunctor(a', b') \).
    If \( a \cong a' \) and \( b \cong b' \), there exist \( \alpha \in \TGrp(a',a) \) and \( \beta \in \TGrpb(b,b') \) such that \( x' = \alpha \cdot x \cdot \beta \).
\end{lemma}
\begin{proof}
    Take any \( \alpha_0 \in \TGrp(a', a) \) and \( \beta_0 \in \TGrp(b,b') \).
    By considering \( \alpha_0 \cdot x \cdot \beta_0 \) instead of \( x \), we can assume without loss of generality that \( a = a' \) and \( b = b' \).

    Let \( x, y \in \Profunctor(a,b) \).
There exist \( T \in \TotalElement(\TGrp) \) and \( C \in \CototalElement(\TGrpb) \) such that \( T(a) \neq \emptyset \) and \( C(b) \neq \emptyset \).
    Let \( z \in T(a) \) and \( u \in C(b) \).
    By the totality of \( \Profunctor \), we have \( (T; \Profunctor) \in \TotalElement(\TGrpb) \), so \( (T; \Profunctor) \TotalOrthogonal C \) and hence \( (T ; \Profunctor ; C) \cong \ident_I \).
    The set \( (T ; \Profunctor ; C)(\star, \star) \) contains the equivalence classes of \( (z, x, u) \) and \( (z, y, u) \), which must coincide since \( (T ; \Profunctor ; C) \cong \ident_I \).
    They belong to the same equivalence class if and only if \( (z, y, u) = (z \cdot \alpha^{-1}, \alpha \cdot x \cdot \beta, \beta^{-1} \cdot u) \) for some \( \alpha \) and \( \beta \).
    In particular, \( y = \alpha \cdot x \cdot \beta \).
\end{proof}
So a stably total profunctor \( \Profunctor \in \TotProf(\TGrp, \TGrpb) \) is determined by
\begin{itemize}
    \item the totality map \( \Truncation{\Profunctor} \in \Tot(\Truncation{\TGrp}, \Truncation{\TGrpb}) \) that describes which pair \( (a,b) \) of objects are related, and
    \item for each pair \( (a,b) \) with \( \Profunctor(a,b) \neq \emptyset \), the conjugation class of the stabiliser group \( \Stabiliser(x) \) of an element \( x \in \Profunctorb(a,b) \). 
\end{itemize}

A maximality property also holds for stably total profunctors.
We consider the one-sided case for simplicity.
We write \( \TUnit \) for the tensor unit in \( \TotProf \), which is the trivial groupoid \( I \) equipped with the unique totality structure.
\begin{lemma}\label{lem:total:equivalence-principle}
    Let \( \TGrp \) be a totality groupoid, \( \Profunctor \colon I \profarrow \Underlying(\TGrp) \) be a profunctor, and \( \Profunctorb \colon \TUnit \profarrow \TGrp \) be a stably total profunctor.
    Assume the following conditions.
    \begin{itemize}
        \item The underlying relation of\/ \( \Profunctor \) is total, i.e., \( \Truncation{\Profunctor} \in \Tot(I, \Truncation{\TGrp}) \).
        \item \( \Profunctor \) is \emph{weakly total} in the sense that, for every \( a \in \Grp \) with \( \Profunctor(\star, a) \neq \emptyset \), there exists \( C \in \CototalElement(\TGrp) \) such that \( C(a, \star) \neq \emptyset \) and \( \Profunctor \TotalOrthogonal C \).
        \item For every \( a \in \Grp \) with \( \Profunctor(\star, a) \neq \emptyset \), there exist \( a' \cong a \), \( x \in \Profunctor(\star,a') \) and \( y \in \Profunctorb(\star,a') \) such that \( \Stabiliser(x) \subseteq \Stabiliser(y) \).
\end{itemize}
    Then \( \Profunctor \cong \Profunctorb \).
\end{lemma}
\begin{proof}
We choose a representative for each isomorphism class of objects in \( \Grpb \).
For each representative \( a \in \Grp \) with \( \Profunctor(\star,a) \neq \emptyset \), we find a pair \( (x_{a}, y_{a}) \in \Profunctor(\star,a) \times \Profunctorb(\star,a) \) such that \( \Stabiliser(x_{a}) = \Stabiliser(y_{a}) \) as follows.
    Let \( a' \), \( x \) and \( y \) be those in the third condition.
    By the second condition, there is a cototal profunctor \( C \in \CototalElement(\TGrp) \) such that \( C(a', \star) \neq \emptyset \) and \( \Profunctor \TotalOrthogonal C \).
    We also have \( \Profunctorb \TotalOrthogonal C \) since \( \Profunctorb \) is stably total.
    Take any \( z \in C(a', \star) \).
    Since \( \Profunctor \TotalOrthogonal C \), it must be the case that \( \Stabiliser(x) \TotalOrthogonal \Stabiliser(z) \), and similarly \( \Stabiliser(y) \TotalOrthogonal \Stabiliser(z) \).
    Hence, the assumption \( \Stabiliser(x) \subseteq \Stabiliser(y) \) implies \( \Stabiliser(x) = \Stabiliser(y) \) by \cref{lem:totality:basic-property-of-total-orthogonality-on-groups}.
    Take an arbitrary morphism \( \beta \in \Grpb(b', b) \) and let \( x_b \defeq x \cdot \beta \) and \( y_b \defeq y \cdot \beta \).

    We give a natural transformation \( \psi \colon \Profunctor \Rightarrow \Profunctorb \) as follows.
    Let \( x' \in \Profunctor(\star, b') \) and \( b \) be the representative of the isomorphism class of \( b' \).
    By \cref{lem:total:single-orbit-property}, there exists \( \beta \in \Grpb(b,b') \) such that \( x' = x_{b} \cdot \beta \).
    Then we define \( \psi_{\star,b'}(x') \defeq (y_{b} \cdot \beta) \).
    The choice of \( \beta \) is not unique, but this is well-defined since \( \Stabiliser(x_{b}) \subseteq \Stabiliser(y_{b}) \).
    This is natural by construction and injective since \( \Stabiliser(x_{b}) \supseteq \Stabiliser(y_{b}) \).

    We prove that \( \psi_{\star,b'} \) is surjective.
    Note that the condition implies the containment \( \Truncation{\Profunctor} \subseteq \Truncation{\Profunctorb} \) between the associated relations.
    By the maximality property for totality spaces (\cref{lem:total:loader-maximality}), we have \( \Truncation{\Profunctor} = \Truncation{\Profunctorb} \) (note that \( \Truncation{\Profunctor} \) is total by the first condition).
    Assume \( y' \in \Profunctorb(\star,b') \).
    The above argument shows that \( \Profunctor(\star,b') \neq \emptyset \), so \( x_{b} \) and \( y_{b} \) are defined where \( b \) is the representative of the isomorphism class of \( b' \).
    By the single orbit property (\cref{lem:total:single-orbit-property}), \( y' = y_{b} \cdot \beta \) for some \( \beta \).
    So \( y' = \psi_{\star,b'}(x_{b} \cdot \beta) \).
\end{proof}

 \section{Full Definability}
Every proof structure induces a family of profunctors parameterized by a groupoid assignment to propositional variables.
This section provides a complete characterization of families definable by \( \mathsf{MLL+Mix} \) proofs based on the category \( \TotProf \) of stably total profunctors.

\subsection{Logical Characterization of Definability}
\newcommand{\TwoCell}{a}
The interpretation \( \sem{\pi}_{\vec{\Grp}} \) of a proof \( \pi \vdash \Formula \) is parameterized by the assignment \( \vec{\Grp} \) of groupoids to propositional variables.
The family \( (\sem{\pi}_{\vec{\Grp}})_{\vec{\Grp}} \) is expected to exhibit similar structures across different components, and the technical challenge lies in how to capture the relationships between components.
A particular obstacle is the presence of negative occurrences of propositional variables in \( \Formula \).
Our approach is to use functors to describe the relationships between components, leveraging the fact that the opposition \( ({-})^{\op} \), the negation in \( \Prof \), is a \emph{covariant} functor.

To accommodate both profunctors and functors within a single framework, we use the \emph{double category of profunctors}.
This paper does not rely on any deep results from double category theory; we merely introduce and use the following notation.
For profunctors \( \Profunctor_1 \colon \Grp_1 \profarrow \Grpb_1 \) and \( \Profunctor_2 \colon \Grp_2 \profarrow \Grpb_2 \) and functors \( \Functor \colon \Grp_1 \longrightarrow \Grp_2 \) and \( \Functorb \colon \Grpb \longrightarrow \Grpb_2 \), we write
\begin{equation*}
    \xymatrix{
        \Grp_1 \ar[d]^{\Functor} \ar[rr]|{|}^{\Profunctor_1} \ar@{}[drr]|{\Downarrow} & & \Grpb_1 \ar[d]_{\Functorb} \\
        \Grp_2 \ar[rr]|{|}_{\Profunctor_2} & & \Grpb_2
    }
\end{equation*}
to mean that there exists a natural transformation \( \Profunctor_1({-}, {+}) \Longrightarrow \Profunctor_2(\Functor({-}), \Functorb({+})) \).
(Note that the existence of a natural transformation is well-defined on isomorphisms classes of profunctors.)

An \( \mathsf{MLL} \) formula \( \Formula \) with free variables \( \PropVar_1,\dots,\PropVar_N \) defines a mapping from sequences \( \vec{\Grp} = (\Grp_1,\dots,\Grp_n) \) to groupoids \( \sem{\Formula}_{\vec{\Grp}} \), which we shall simply write as \( \Formula(\vec{\Grp}) \).
Similarly, it determines a totality groupoid \( \Formula(\vec{\TGrp}) \defeq \sem{\Formula}_{\vec{\TGrp}} \) for each \( n \)-tuple \( \vec{\TGrp} = (\TGrp_1,\dots,\TGrp_n) \) of totality groupoids, and we have \( \Underlying(\Formula(\vec{\TGrp})) = \Formula(\Underlying(\vec{\TGrp})) \).

Importantly, it also determines a functor \( \Formula(\Phi) \colon \Formula(\vec{\Grp}) \longrightarrow \Formula(\vec{\Grpb}) \) for each \( n \)-tuple \( \Phi = (\varphi_1,\dots,\varphi_n) \) of functors \( \varphi_i \colon \Grp_i \longrightarrow \Grpb_i \).
It is simply defined by \( (\Formula \otimes \Formulab)(\Phi) = (\Formula \llpar \Formulab)(\Phi) = \Formula(\Phi) \times \Formulab(\Phi) \) and \( \Formula^\bot(\Phi) = (\Formula(\Phi))^{\op} \).
\begin{definition}[Logical family of profunctors]\label{def:logical-family}
    A family \( \varpi = (\varpi_{\vec{\Grp}})_{\vec{\Grp}} \) of profunctors \( \varpi_{\vec{\Grp}} \colon I \profarrow \Formula(\vec{\Grp}) \) is \emph{logical} if,
    for every \( \Phi \colon \vec{\Grp} \longrightarrow \vec{\Grpb} \),
    \begin{equation*}
        \xymatrix{
            I \ar@{=}[d] \ar[rr]|{|}^{\varpi_{\vec{\Grp}}} \ar@{}[drr]|{\Downarrow} & & \Formula(\vec{\Grp}) \ar[d]^{\Formula(\Phi)} \\
            I \ar[rr]|{|}_{\varpi_{\vec{\Grpb}}} & & \Formula(\vec{\Grpb})
        }
    \end{equation*}
    holds.
    It is \emph{stably total} if every component is stably total, i.e.~\( \varpi_{\vec{\TGrp}} \defeq \varpi_{\Underlying(\vec{\TGrp})} \in \TotProf(I, \Formula(\vec{\TGrp})) \) for every tuple \( \vec{\TGrp} \) of totality groupoids.

    In some cases, we consider a family in which the parameter \( \vec{\Grp} \) takes values in (tuples of) a subclass of groupoids, such as groups (i.e.~single-object groupoids) or sets (i.e.~groupoids in which all morphisms are identities).
    The logicality in these cases is defined similarly, considering only \( \Phi \colon \vec{\Grp} \longrightarrow \vec{\Grpb} \) between the subclass.
\qed
\end{definition}

The next theorem is the main result of this paper.
The proof will be presented in the following subsection.
\begin{theorem}\label{thm:definability}
    Let \( \Formula \) be an MLL formula and \( \varpi = (\varpi_{\vec{\Grp}} \colon I \profarrow \Formula(\vec{\Grp}))_{\vec{\Grp}} \) be a family of profunctors.
    It is definable (up to natural isomorphism) by an \( \mathsf{MLL}+\mathsf{Mix} \) proof if and only if it is logical and stably total.
\qed
\end{theorem}

\subsection{Proof of Full Definability}
This subsection proves \cref{thm:definability}.
One direction is easy.
\begin{lemma}\label{lem:definability:easy-direction}
    For an \( \mathsf{MLL+Mix} \) proof \( \pi \vdash \Formula \), the family \( (\sem{\pi}_{\vec{\Grp}} \colon I \profarrow \Formula(\vec{\Grp}))_{\vec{\Grp}} \) is logical and stably total.
\end{lemma}
\begin{proof}
    \( \sem{\pi}_{\vec{\Grp}} \) is stably total since \( \TotProf \) is a model of \( \mathsf{MLL+Mix} \) and \( \Underlying \) strictly preserves the \( \mathsf{MLL}+\mathsf{Mix} \) structure.
    The logicality of the family can be easily proved by induction on proofs.
\end{proof}

\begin{remark}
    The logicality of the interpretation in the above proof corresponds to the fact that each logical connective is a (double) functor on the double category.
    In this sense, the structure of the double category is used indirectly in the paper.
    Such a use of double categories subsumes the classical approach based on embedding-projection, but we will not go further into the details here.
\end{remark}

We prove the converse.
Let \( \Formula \) be an \( \mathsf{MLL} \) formula and \( \varpi = (\varpi_{\vec{\Grp}} \colon I \profarrow \Formula(\vec{\Grp}))_{\vec{\Grp}} \) be a logical and stably total family.
Identifying \( \varpi_{\vec{\Grp}} \colon I \profarrow \Formula(\vec{\Grp}) \) with the covariant presheaf \( \varpi_{\vec{\Grp}} \colon \Formula(\vec{\Grp}) \longrightarrow \Set \), we write \( \varpi_{\vec{\Grp}}(a_1,\dots,a_n) \) for \( \varpi_{\vec{\Grp}}(\star, (a_1,\dots,a_n)) \).

The proof proceeds as follows:
\begin{enumerate}
    \item Prove that the family \( \varpi \) is an interpretation of a proof structure \( \pi \).
    \item Prove that the proof structure \( \pi \) is correct.
\end{enumerate}
The latter has already been addressed: since the interpretation \( \varpi \) of \( \pi \) is stable by the assumption, \cref{thm:creed-kit-correctness-without-cut} shows that \( \pi \) is correct.
So this section mainly focuses on the former.

The proof of (1) consists of two parts:
\begin{enumerate}
    \item[(1a)] Solve the relational version, i.e., construct a proof structure \( \pi \) whose relational interpretation coincides with \( \Truncation{\varpi} \).
    \item[(1b)] Then extend (1a) to the profunctorial setting, i.e., prove that \( \varpi \) is the profunctorial interpretation of \( \pi \).
\end{enumerate}
The step (1a) has essentially been given by Loader~\cite{Loader1994}, which proves the full definability of the totality space model, because the category \( \Tot \) of totality spaces is a full subcategory of \( \TotProf \) (consisting of totality groupoids whose underlying groupoids are setoids).
So the family \( \varpi \) has a logical family parameterized by totality spaces, and we can utilize (a variant of) Loader's result on totality spaces~\cite{Loader1994} to extract a proof structure.
\begin{lemma}\label{lem:definability:loader-linking}
    There exists a proof structure \( \pi \) of \( \Formula \) such that
    \begin{equation*}
        \Truncation{\varpi_{\vec{\Grp}}} = \Truncation{\sem{\pi}_{\vec{\Grp}}}.
    \end{equation*}
    for every \( \vec{\Grp} \).
\end{lemma}
\begin{proof}
    Since \( \Tot \hookrightarrow \TotProf \), the family \( \varpi \) contains a logical family parameterized by totality spaces, so the result is closely related to \cite[Lemma~8]{Loader1994}.
    A subtlety here is that the logicality of this paper is slightly weaker than the uniformity by Loader~\cite[Definition~2]{Loader1994}.
    For \( \Phi \colon \vec{\Grp} \longrightarrow \vec{\Grpb} \), our logicality requires only \( \varpi_{\vec{\Grp}}(a) \neq \emptyset \Longrightarrow \varpi_{\vec{\Grpb}}(\Formula(\Phi)(a)) \neq \emptyset \), but Loader's uniformity also requires a certain form of the converse.
    However, careful examination reveals that the additional requirement is not needed to prove \cite[Lemma~8]{Loader1994} (see \cref{sec:appx:loader} for a complete proof).

    Loader's lemma gives \( \pi \) and proves the above equality when \( \vec{\Grp} \) are sets (or setoids), say \( \vec{\Grp} = \Truncation{\vec{\Grpb}} \).
    To prove the equality for all groupoids, note that \( \Set \hookrightarrow \GroupoidCat \) is reflective and the unit \( \Grp \longrightarrow \Truncation{\Grp} \) of the reflection is essentially bijective on objects.
    So \( \Truncation{\varpi_{\vec{\Grp}}} = \Truncation{\varpi_{\Truncation{\vec{\Grp}}}} = \Truncation{\sem{\pi}_{\Truncation{\vec{\Grp}}}} = \Truncation{\sem{\pi}_{\vec{\Grp}}} \) by the logicality of the family.
\end{proof}

Let \( \pi \) be the proof structure in \cref{lem:definability:loader-linking}, fixed in the sequel.
Let \( \dot{\Formula}(\PropVarc_1, \PropVarb_1,\dots,\PropVarc_n, \PropVarb_n) \) be a formula obtained by replacing each occurrence of a propositional variable with a fresh variable (where \( \PropVarb_i \) and \( \PropVarc_i \) are positive and negative occurrences).
We can assume without loss of generality that \( \PropVarb_i \) and \( \PropVarc_i \) are connected by an axiom link in \( \pi \).
We have \( \Formula(\PropVar_1,\dots,\PropVar_N) = \dot{\Formula}(\PropVar_{\xi_1}, \PropVar_{\xi_1}, \dots, \PropVar_{\xi_n}, \PropVar_{\xi_n}) \) for some \( \xi_1,\dots,\xi_n \in \{ 1, \dots, N \} \).

\begin{example}
    For \( \Formula = ((\PropVar_1 \llpar \PropVar_2) \otimes \PropVar_2^{\bot}) \llpar ((\PropVar_2 \llpar \PropVar_2^\bot) \llpar \PropVar_1^{\bot}) \), we can choose \( \dot{\Formula}(\PropVarc_1, \PropVarb_1, \PropVarc_2, \PropVarb_2, \PropVarc_3, \PropVarb_3) = ((\PropVarb_2 \llpar \PropVarb_1) \otimes \PropVarc_3^{\bot}) \llpar ((\PropVarb_3 \llpar \PropVarc_1^\bot) \llpar \PropVarc_2^{\bot}) \).
    Then \( \Formula = \dot{\Formula}(\PropVar_2, \PropVar_2, \PropVar_1, \PropVar_1, \PropVar_2, \PropVar_2) \), so \( \xi_1 = 2 \), \( \xi_2 = 1 \) and \( \xi_3 = 2 \).
\qed
\end{example}

As discussed in \cref{sec:pre:interpretation}, there exists a family of groupoid isomorphisms
\begin{align*}
    A(\vec{\Grp}) \:\stackrel{\Functor_{\vec{\Grp}}^{-1}}{\cong}\: \Grp_{\xi_1}^{\op} \times \Grp_{\xi_1} \times \dots \times \Grp_{\xi_n}^{\op} \times \Grp_{\xi_n}
\end{align*}
parameterized by \( \vec{\Grp} \).
Let \( \varpi' \colon I \profarrow \Grp_{\xi_1}^{\op} \times \Grp_{\xi_1} \times \dots \times \Grp_{\xi_n}^{\op} \times \Grp_{\xi_n} \) be a profunctor given by
\begin{equation*}
    \varpi'_{\vec{\Grp}}(\star, (c_1, b_1, \dots, c_n, b_n)) \defeq \varpi_{\vec{\Grp}}(\star, \Functor_{\vec{\Grp}}(c_1, b_1, \dots, c_n, b_n)).
\end{equation*}
Since
\begin{equation*}
    \sem{\pi}_{\vec{\Grp}}(\star, \Functor_{\vec{\Grp}}(c_1, b_1, \dots, c_n, b_n)) \cong \Grp_{\xi_1}(c_1, b_1) \times \dots \times \Grp_{\xi_n}(c_n, b_n),
\end{equation*}
The step (1b) is to show that
\begin{equation*}
    \varpi'_{\vec{\Grp}}(\star, (c_1, b_1, \dots, c_n, b_n)) \cong \Grp_{\xi_1}(c_1, b_1) \times \dots \times \Grp_{\xi_n}(c_n, b_n)
\end{equation*}
natural in \( c_1, b_1, \dots, c_n, b_n \) for every \( \vec{\Grp} \).
We shall often omit \( \star \) and simply write \( \varpi'_{\vec{\Grp}}(c_1, b_1, \dots, c_n, b_n) \) for \( \varpi'_{\vec{\Grp}}(\star, (c_1, b_1, \dots, c_n, b_n)) \).

We first consider a simple case where \( n = 1 \) and \( \vec{\Grp} \) are groups.
Even this simple case needs a non-trivial result from group theory.
A group isomorphism \( \varphi \colon G \stackrel{\cong}{\longrightarrow} G \) on a group \( G \) is called an \emph{automorphism}.
For example, an element \( g \in G \) induces an automorphism \( h \mapsto g h g^{-1} \).
An automorphism of this form is called an \emph{inner automorphism}; an automorphism is called an \emph{outer automorphism} if it is not an inner automorphism.
We use the following characterization of inner automorphisms.
\begin{theorem}[Schupp~\cite{Schupp1987}]\label{thm:definability:characterization-of-inner}
    Let \( G \) be a group and \( \varphi \colon G \stackrel{\cong}{\longrightarrow} G \) is a group automorphism.
    Then \( \varphi \) is an inner automorphism if and only if \( \varphi \) can be extended to an automorphism via any group injection, i.e.,
    \begin{quote}
        for every injective group homomorphism \( \iota \colon G \hookrightarrow H \), there exists an automorphism \( \psi \colon H \stackrel{\cong}{\longrightarrow} H \) such that \( \psi(\iota(g)) = \iota(\varphi(g)) \) for every \( g \in G \).
    \end{quote}
\end{theorem}

\begin{lemma}\label{lem:definability:inner-automorphism}
    Let \( \chi = (\chi_{\Group} \colon I \profarrow \Group^{\op} \times \Group)_{\Group} \) be a logical family of profunctors parameterized by groups \( \Group \) (regarded as single-object groupoids).
    Assume that
    \begin{align*}
        \chi_{\Group} \TotalOrthogonal \big(\InducedProf{\{\ident\}} \times \InducedProf{\Group}\big)
        \quad\mbox{and}\quad
        \chi_{\Group} \TotalOrthogonal \big(\InducedProf{\Group^{\op}} \times \InducedProf{\{\ident\}}\big)
\end{align*}
    for every \( \Group \).
    Then \( \chi_{\Group}(\star, ({-}, {+})) \cong \Group({-}, {+}) \).
\end{lemma}
\begin{proof}
    By \( \chi_{\Group} \TotalOrthogonal \big(\InducedProf{\{\ident\}} \times \InducedProf{\Group}\big) \), we know that \( \chi_{\Group} \) has a single orbit by \cref{lem:total:single-orbit-property}.
    Let \( x \in \chi_{\Group}(\star, (\star,\star)) \) and \( \Groupb = \Stabiliser(x) \le (\Group^{\op} \times \Group) \).
    Then \( \chi_{\Group} \cong \InducedProf{\Groupb} \) (because \( \chi_{\Group} \) has a single orbit).

    Since \( \InducedProf{\{\ident\}} \times \InducedProf{\Group} \cong \big(\InducedProf{\{\ident\}\times \Group}\big) \), the stably total orthogonality means that \( (\Groupb, \{\ident\}\times \Group) \) forms a strict factorization system (\cref{prop:totality:strict-factorization-system}).
    So \( (\alpha, \ident) \in \Group^{\op} \times \Group \) can be uniquely factorized as \( (\alpha, \ident) = (\alpha_1, \alpha_2) (\ident, \alpha'_2) \) with \( (\alpha_1,\alpha_2) \in \Groupb \).
    By the definition of the composition in \( \Group^{\op} \times \Group \), we know \( \alpha_1 = \alpha \), so \( (\alpha, \alpha_2) \in \Groupb \).
    Furthermore, this \( \alpha_2 \) is a unique element related to \( \alpha \): if \( (\alpha, \alpha_2), (\alpha, \alpha_2') \in \Groupb \), then \( (\alpha, \alpha_2)^{-1} (\alpha, \alpha_2') = (\ident, \alpha_2^{-1} \alpha_2') \in \Groupb \), so \( \alpha_2^{-1} \alpha_2' = \ident \) since \( (\ident, \alpha_2^{-1} \alpha_2') \in (\Groupb \cap (\{\ident\} \times \Group)) \).
    This observation holds for every \( \alpha \), so \( \forall \alpha. \exists! \alpha_2. (\alpha, \alpha_2) \in \Groupb \).
    Let \( \vartheta \colon \Group \longrightarrow \Group \) be the function defined by \( (\alpha^{-1}, \vartheta(\alpha)) \in \Groupb \).

    The function \( \vartheta \) is a group homomorphism.
    Actually, since \( \Groupb \) is a subgroup of \( \Group^{\op} \times \Group \), we know that \( \Groupb \) contains
    \begin{align*}
        (\alpha_1^{-1}, \vartheta(\alpha_1)) (\alpha_2^{-1}, \vartheta(\alpha_2))
        &= (\alpha_2^{-1} \alpha_1^{-1}, \vartheta(\alpha_1)\vartheta(\alpha_2)),
        \\
        &= ((\alpha_1\alpha_2)^{-1}, \vartheta(\alpha_1)\vartheta(\alpha_2)),
    \end{align*}
    so \( \vartheta(\alpha_1\alpha_2) = \vartheta(\alpha_1)\vartheta(\alpha_2) \).

    The function \( \vartheta \) is a group automorphism.
    We prove this claim by constructing the inverse homomorphism.
    A similar argument using \( \chi_{\Group} \TotalOrthogonal \big(\InducedProf{\Group^{\op}} \times \InducedProf{\{\ident\}}\big) \) gives a group homomorphism \( \vartheta' \colon \Group \longrightarrow \Group \) satisfying \( (\vartheta'(\beta), \beta^{-1}) \in \Groupb \) for every \( \beta \in \Group \).
    Then
    \begin{equation*}
         (\alpha^{-1}, \vartheta(\alpha)) (\vartheta'(\beta), \beta^{-1})
         \;=\; (\vartheta'(\beta)\,\alpha^{-1}, \vartheta(\alpha)\,\beta^{-1})
         \;\in\; \Groupb 
    \end{equation*}
    for every \( \alpha, \beta \in \Group \).
    By instantiating \( \beta = \vartheta(\alpha) \), we have
    \begin{equation*}
        (\vartheta'(\vartheta(\alpha))\,\alpha^{-1} , \ident) \;\in\; \Groupb,
    \end{equation*}
    so \( \alpha = \vartheta'(\vartheta(\alpha)) \) since \( (\vartheta'(\vartheta(\alpha))\,\alpha^{-1}, \ident) \in (\Groupb \cap (\Group^{\op} \times \{\ident\})) \).
    A similar argument shows that \( \beta = \vartheta(\vartheta'(\beta)) \).

    Note that the construction of the group automorphism \( \vartheta \) is parameterized by group \( \Group \) and element \( x \in \chi_G(\star, \star) \).
    We shall write \( \vartheta_{\Group, x} \) for the above-given automorphism to make the parameters \( \Group \) and \( x \) explicit.

    We prove that \( \vartheta_{\Group, x} \) is an inner automorphism, appealing to \cref{thm:definability:characterization-of-inner}.
    Let \( \iota \colon \Group \hookrightarrow \Group' \) be an arbitrary group embedding.
    Since \( \chi \) is logical, we have a function \( \psi \colon \chi_{\Group}(\star,(\star,\star)) \longrightarrow \chi_{\Group'}(\star,(\star,\star)) \) such that \( \psi(y \cdot (\alpha_1^{-1}, \alpha_2)) = \psi(y) \cdot (\iota(\alpha_1)^{-1}, \iota(\alpha_2)) \) for every \( y \in \chi_{\Group}(\star,(\star,\star)) \) and  \( \alpha_1, \alpha_2 \in \Group \).
    For \( \alpha \in \Group \), since \( x \cdot (\alpha^{-1}, \vartheta_{\Group, x}(\alpha)) = x \), we have \( \psi(x) \cdot (\iota(\alpha)^{-1}, \iota(\vartheta_{\Group, x}(\alpha))) = \psi(x) \), so \( (\iota(\alpha)^{-1}, \iota(\vartheta_{\Group, x}(\alpha))) \in \Stabiliser(\psi(x)) \).
    This means that \( \vartheta_{\Group', \psi(x)}(\iota(\alpha)) = \iota(\vartheta_{\Group, x}(\alpha)) \).
    Since \( \alpha \in \Group \) is arbitrary, this means that the group automorphism \( \vartheta_{\Group, x} \) on \( \Group \) can be extended to a group automorphism \( \vartheta_{\Group', \psi(x)} \) on \( \Group' \) via the embedding \( \iota \).
    Since \( \iota \) is arbitrary, \( \vartheta_{\Group, x} \) satisfies the condition in \cref{thm:definability:characterization-of-inner}.
    Therefore, \( \vartheta_{\Group, x} \) is an inner automorphism.

    Hence, \( \vartheta(\alpha) = \beta \alpha \beta^{-1} \) for some \( \beta \in \Group \).
    Let \( x' = x \cdot (\ident, \beta^{-1}) \).
Then the automorphism induced by \( \Stabiliser(x') \) is the identity, i.e.~\( \Stabiliser(x') = \{ (\alpha^{-1}, \alpha) \mid \alpha \in \Group \} \).
    So \( \chi_{\Group} \) is isomorphic to the transpose of the identity profunctor.
\end{proof}

We then prove the general case by using the above result and \cref{lem:appx:totality:equality-condition}.
The next lemma gives a lower bound of the stabiliser.
\begin{lemma}\label{lem:definability:stabiliser-contains-innor-automorphism}
    Assume \( x \in \varpi'_{\vec{\Grp}}(b_1, b_1, \dots, b_n, b_n) \).
    For each \( 1 \le j \le n \), there exists an inner automorphism \( \varphi_j \) on \( \Grp_{\xi_j}(b_j, b_j) \) such that, for every \( \alpha \in \Grp_{\xi_j}(b_j,b_j) \),
    \begin{equation*}
        (\overbrace{\ident, \dots, \ident}^{2(j-1)}, \alpha^{-1}, \varphi_j(\alpha), \overbrace{\ident, \dots, \ident}^{2(N-j)}) \quad\in\quad \Stabiliser(x).
    \end{equation*}
\end{lemma}
\begin{proof}
    We can assume without loss of generality that \( j = 1 \) and \( i_j = 1 \).
    For a group \( \Group \), let \( \mathbb{D}_{\Group} = \Group + \overbrace{I + \dots + I}^{n-1} \) (the coproduct as groupoids) and \( \vec{\mathbb{D}}_{\Group} = (\mathbb{D}_{\Group},\dots,\mathbb{D}_{\Group}) \).
    We write \( \star_j \) for the \( j \)-th object in \( \mathbb{D}_{\Group} \) and define a profunctor \( \chi_{\Group} \colon I \profarrow \Group^{\op} \times \Group \) by
    \begin{equation*}
        \chi_{\Group}({-}, {+})
        \quad\defeq\quad
        \varpi_{\vec{\mathbb{D}}_{\Group}}(\mathsf{inj}_1({-}),\mathsf{inj}_1({+}),\star_2,\star_2,\dots,\star_n,\star_n)
    \end{equation*}
    where \( \mathsf{inj}_1 \colon G \longrightarrow \mathbb{D}_{\Group} \).
    We write the unique object of \( G \) as \( \ast \) to distinguish \( \star \in I \).
    Let \( \mathcal{G}^{+} \) and \( \mathcal{G}^{-} \) be the totality structure on \( \Group \) such that \( \mathcal{G}^+_{\ast} = \{ \Group \} \) and \( \mathcal{G}^{-}_\ast = \{ \{ \ident \} \} \),
    \( \mathcal{D}^{\pm}_\Group = \mathcal{G}^{\pm} + I + \dots + I \),
    and \( \vec{\mathcal{D}}^{\pm}_{\Group} = (\mathcal{D}^{\pm}_\Group, \dots, \mathcal{D}^{\pm}_\Group) \).
    Then \( \Group \times \{ \ident \} \times 1 \times \dots \times 1 \) belongs to \( \Formula(\vec{\mathcal{D}}^+)^{\bot} \) and \( \{ \ident \} \times \Group \times 1 \times \dots \times 1 \) belongs to \( \Formula(\vec{\mathcal{D}}^-)^{\bot} \) (where \( 1 = I(\star,\star) \) is the trivial group).
    Since \( \varpi \) is total, we have
    \begin{align*}
        \chi_{\Group} \TotalOrthogonal \big(\InducedProf{\Group} \times \InducedProf{\{ \ident \}}\big)
        \\
        \chi_{\Group} \TotalOrthogonal \big(\InducedProf{\{ \ident \}} \times \InducedProf{\Group}\big)
    \end{align*}

    Since \( \Group \) is arbitrary, we obtain a family \( \chi = (\chi_{\Group})_{\Group} \) parameterized by group \( \Group \).
    This family is logical since \( \varpi \) is.
    By \cref{lem:definability:inner-automorphism}, \( \chi_{\Group} \) is isomorphic to the transpose of the identity profunctor.
    So there exists \( x_{\Group} \in \chi_{\Group}(\ast,\ast) \) such that \( \Stabiliser(x_{\Group}) = \{ (\alpha^{-1}, \alpha) \mid \alpha \in \Group \} \).

    Let \( \Group = \Grp_1(b_1,b_1) \) and consider functors \( \varphi_i \colon \mathbb{D}_{\Group} \longrightarrow \Grp_i \) that satisfy \( \varphi_{i_j}(\star_j) = b_j \) for every \( j = 1,\dots,n \) and \( \varphi_1(\mathsf{inj}_1(\alpha)) = \alpha \) for every \( \alpha \in \Grp_1(b_1,b_1) \).
    Let \( \Phi = (\varphi_1,\dots,\varphi_m) \).
    By the logicality, there exists a natural transformation \( \varpi_{\vec{\mathbb{D}}_{\Group}} \Longrightarrow (\Formula(\Phi);\varpi_{\vec{\Grp}}) \).
    Let \( x'_{\Group} \) be the image of \( x_\Group \) by the natural transformation.
    Then \( (\alpha^{-1}, \alpha, \ident, \dots, \ident) \in \Stabiliser(x'_{\Group}) \) for every \( \alpha \in \Grp_1(b_1,b_1) \).
    Since \( \varpi_{\vec{\Grp}} \) is total, \( \varpi_{\vec{\Grp}}(b_1,b_1,\dots,b_n,b_n) \) has a single orbit, so \( x = x'_{\Group} \cdot \beta \) for some \( \beta = (\beta_1, \beta_1', \dots, \beta_n, \beta_n') \in \Formula(\vec{\Grp})((b_1,b_1,\dots,b_n,b_n), (b_1,b_1,\dots,b_n,b_n)) \).
    Then \( \varphi_1(\alpha) = (\beta_1' \beta_1) \alpha (\beta_1' \beta_1)^{-1} \) satisfies the condition.
\end{proof}

\begin{corollary}\label{cor:definability:identity-stabilizer}
    Assume that \( \varpi'_{\vec{\Grp}}(b_1,b_1,\dots,b_n,b_n) \neq \emptyset \).
    There exists \( x \in \varpi'_{\vec{\Grp}}(b_1,b_1,\dots,b_n,b_n) \) such that
    \begin{equation*}
        (\alpha_1^{-1}, \alpha_1, \dots, \alpha_n^{-1}, \alpha_n) \in \Stabiliser(x)
    \end{equation*}
    for every \( (\alpha_1,\dots,\alpha_n) \in \prod_{i = 1}^n \Grp_{\xi_i}(b_i, b_i)\).
\end{corollary}
\begin{proof}
    Suppose \( y \in \varpi'_{\vec{\Grp}}(b_1,b_1,\dots,b_n,b_n) \).
    By \cref{lem:definability:stabiliser-contains-innor-automorphism}, for each \( j \), there is an inner automorphism \( \varphi_j \) on \( \Grp_{\xi_j}(b_j, b_j) \) such that
    \begin{equation*}
        (\ident, \dots, \ident, \alpha_j^{-1}, \varphi_j(\alpha_j), \ident, \dots, \ident) \in \Stabiliser(y)
    \end{equation*}
    for every \( \alpha_j \in \Grp_{\xi_j}(b_j,b_j) \).
    Since \( \Stabiliser(x) \) is closed under composition,
    \begin{equation*}
        \{\, (\alpha_1^{-1}, \varphi_1(\alpha_1), \dots, \alpha_n^{-1}, \varphi_n(\alpha_n)) \mid \vec{\alpha} \,\} \:\subseteq\: \Stabiliser(y)
    \end{equation*}
    where \( \vec{\alpha} = (\alpha_1,\dots,\alpha_n) \) ranges over \( \Grp_{\xi_1}(b_1,b_1) \times \dots \times \Grp_{\xi_j}(b_j,b_j) \).
    Since \( \varphi_i \) is an inner automorphism, there exists \( \beta_i \) such that \( \varphi_i(\alpha) = \beta_i \alpha \beta_i^{-1} \).
    Hence
    \begin{equation*}
        \{\, (\alpha_1^{-1}, \beta_1\alpha_1\beta_1^{-1}, \dots, \alpha_n^{-1}, \beta_n\alpha_n\beta_n^{-1}) \mid \vec{\alpha} \,\} \:\subseteq\: \Stabiliser(y).
    \end{equation*}
    Let \( x \defeq y \cdot (\ident, \beta_1, \dots, \ident, \beta_n) \).
    Then \( x \cdot (\alpha_1^{-1}, \alpha_1, \dots, \alpha_n^{-1}, \alpha_n) = x \).
\end{proof}

The desired result \( \varpi \cong \sem{\pi} \) is obtained by using the maximality property (\cref{lem:total:equivalence-principle}) to the lower bound given in \cref{cor:definability:identity-stabilizer}.
\begin{lemma}\label{lem:definability:interpretation-of-linking}
    \( \varpi_{\vec{\Grp}} \cong \sem{\pi}_{\vec{\Grp}} \) for every \( \vec{\Grp} \).
\qed
\end{lemma}
\begin{proof}
    Let \( \TGrp_i \) be the \emph{discrete totality} on \( \Grp_i \), that is, \( \TotalElement(\TGrp_i) = \{\, \Grp_i(a, {-}) \mid a \in \Grp_i \,\} \) and \( \CototalElement(\TGrp_i) = \{\, T \,\} \) where \( T(a) = \{ \ast \} \) for every \( a \in \Grp_i \).
    Then \( \CototalElement(A(\vec{\TGrp})) \) contains
    \begin{equation*}
        C_{\vec{a}}(c_1,b_1,\dots,c_n,b_n) \quad=\quad \Grp_{\xi_1}(a_1,c_1) \times \dots \times \Grp_{\xi_n}(a_n,c_n)
    \end{equation*}
    for every sequence \( \vec{a} = (a_1,\dots,a_n) \) of objects \( a_i \in \Grp_{\xi_i} \) (for \( i = 1,\dots,n \)).
    We use \cref{lem:total:equivalence-principle} for a stably total profunctor \( \varpi_{\vec{\TGrp}} \colon \TUnit \profarrow A(\vec{\TGrp}) \) with a (not necessarily stably total) profunctor \( \sem{\pi}_{\vec{\Grp}} \colon I \profarrow A(\vec{\Grp}) \).

    The first condition follows from \( \Truncation{\varpi_{\vec{\Grp}}} = \Truncation{\sem{\pi}_{\vec{\Grp}}} \) and the totality assumption on \( \varpi \).
    The second is a consequence of the fact that \( \sem{\pi}_{\vec{\Grp}}(c_1,b_1,\dots,c_n,b_n) = \Grp_{\xi_1}(c_1,b_1) \times \dots \times \Grp_{\xi_n}(c_n,b_n) \) is orthogonal to \( C_{\vec{a}} \) for every \( \vec{a} \).

    We prove the third condition.
    Assume \( \sem{\pi}_{\vec{\Grp}}(c_1,b_1,\dots,c_n,b_n) \neq \emptyset \).
    Then \( b_j \cong c_j \) for every \( j = 1, \dots, n \).
    We choose \( (b_1,b_1,\dots,b_n,b_n) \) as the representative of the isomorphism class.
    Then
    \begin{equation*}
        x \defeq (\ident_{b_1}, \dots, \ident_{b_n}) \in \sem{\pi}_{\vec{\Grp}}(b_1,b_1,\dots,b_n,b_n),
    \end{equation*}
    and
    \begin{equation*}
        \Stabiliser(x) \:=\: \{\, (\alpha_1^{-1}, \alpha_1, \dots, \alpha_n^{-1}, \alpha_n) \mid \vec{\alpha} \,\}. 
    \end{equation*}
    By \cref{cor:definability:identity-stabilizer}, there exists \( y \in \varpi_{\vec{\Grp}}(b_1,b_1,\dots,b_n,b_n) \) such that \( \Stabiliser(x) \subseteq \Stabiliser(y) \).

    Hence, \cref{lem:total:equivalence-principle} shows the claim.
\end{proof}

\begin{proof}[Proof of \cref{thm:definability}]
    We have seen that the family definable by an \( \mathsf{MLL}+\mathsf{Mix} \) proof is logical and stably total (\cref{lem:definability:easy-direction}).
    We prove the converse.
    Assume a logical and stably total family \( \varpi \).
Let \( \pi \) be the linking in \cref{lem:definability:loader-linking}.
    By \cref{lem:definability:interpretation-of-linking}, the family \( \varpi \) is definable by the proof structure \( \pi \).
    Since every stably total profunctor is stable, \cref{thm:creed-kit-correctness-without-cut} ensures that \( \pi \) is correct.
\end{proof}

\bibliographystyle{ACM-Reference-Format}
\bibliography{jabref}

\appendix
\clearpage

\section{Proof of \cref{prop:creed:equivalence}}
We first recall the claim.
\begin{claim*}
Let \( \Grp \) be a groupoid.
    \begin{enumerate}
        \item A family \( C = (C_a \subseteq \Grp(a,a))_a \) is a creed if and only if
        \begin{enumerate}
\item \( \alpha \in C_a \) implies \( \alpha^k \in C_a \) for every \( k \in \Int \), and
            \item \( \alpha \in C_a \) and \( \beta \in \Grp(a,b) \) implies \( \beta^{-1} ; \alpha ; \beta \in C_b \).
        \end{enumerate}
        \item A family \( K = (K_a \subseteq \{ \Group \mid \Group \le \Grp(a,a) \})_{a} \) is a kit if and only if \( K \) is closed under the conjugation in the sense that, if \( \Group \in K_a \) and \( \beta \in \Grp(a,b) \), then \( \beta^{-1} \cdot \Group \cdot \beta \in K_b \), where \( \beta^{-1} \cdot \Group \cdot \beta = \{ \beta^{-1} ; \alpha ; \beta \mid \alpha \in \Group \} \).
    \end{enumerate}    
\end{claim*}

\quad

(1, \(\Rightarrow\))

Assume that \( C \) be a creed.
So \( C = \CreedOf{\Profunctor} \) for some \( \Profunctor \colon I \profarrow \Grp \).

If \( \alpha \in C_a \), then \( x \cdot \alpha = x \) for some \( x \in \Profunctor(\star, a) \).
Then \( x \cdot \alpha^k = x \), so \( \alpha^k \in C_a \).

Let \( \alpha \in C_a \) and \( \beta \in \Grp(a,b) \).
By \( \alpha \in C_a \), there exists \( x \in \Profunctor(\star,a) \) such that \( x \cdot \alpha = x \).
Let \( y = x \cdot \beta \).
Then
\begin{align*} 
  y \cdot (\beta^{-1};\alpha;\beta) 
  &= x \cdot (\beta; \beta^{-1}; \alpha; \beta)
  \\
  &= (x \cdot \alpha) \cdot \beta
  \\
  &= x \cdot \beta
  \\
  &= y.
\end{align*}
So \( (\beta^{-1};\alpha;\beta) \in C_b \).

\quad

(1, \(\Leftarrow\))

Assume that \( C \) satisfies the conditions (a) and (b).
For each \( a \in \Grp \) and \( \alpha \in C_a \), let \( \Group_{a,\alpha} = \{\, \alpha^k \mid k \in \Int \,\} \) be the subgroup of \( \Grp(a,a) \) generated by \( \alpha \).
We show that \( \InducedProf{\Group_{a,\alpha}} \colon I \profarrow \Grp \) satisfies \( \alpha \in (\CreedOf{\InducedProf{\Group_{a,\alpha}}})_a \) and \( (\CreedOf{\InducedProf{\Group_{a,\alpha}}})_b \subseteq C_b \) for every \( b \in \Grp \).
Then \( \CreedOf{(\coprod_{a \in \Grp, \alpha \in C_a} \InducedProf{\Group_{a,\alpha}})} = C \).

Trivially, \( \alpha \in (\CreedOf{\InducedProf{\Group_{a,\alpha}}})_a \) holds.

We show that \( (\CreedOf{\InducedProf{\Group_{a,\alpha}}})_b \subseteq C_b \) for every \( b \in \Grp \).
Assume that \( \beta \in (\CreedOf{\InducedProf{\Group_{a,\alpha}}})_b \).
By definition, there exists \( x \in \InducedProf{\Group_{a,\alpha}}(\star,b) \) such that \( x \cdot \beta = x \).
By the definition of the induced profunctor, \( x = \Group_{a,\alpha} \cdot \gamma \) for some \( \gamma \).
So \( \Group_{a, \alpha} \cdot \gamma \cdot \beta = \Group_{a, \alpha} \cdot \gamma \).
Hence, \( \Group_{a, \alpha} \cdot \gamma \cdot \beta \cdot \gamma^{-1} = \Group_{a, \alpha} \).
In particular, \( (\gamma;\beta;\gamma^{-1}) \in \Group_{a, \alpha} \), which implies \( (\gamma;\beta;\gamma^{-1}) = \alpha^k \) for some \( k \in \Int \).
Therefore, \( \beta = \gamma^{-1};\alpha^k;\gamma \).
Since \( \alpha \in C_a \), we have \( \alpha^k \in C_a \) by (a) and \( \gamma^{-1};\alpha^k;\gamma \in C_b \) by (b).

\quad

(2, \(\Rightarrow\))

Assume that \( K \) is a kit.
By definition, \( K = \KitOf{\Profunctor} \) for some \( \Profunctor \colon I \profarrow \Grp \).
Assume that \( \Group \in K_a \), i.e., \( \Group = \Stabiliser(x) \) for some \( x \in \Profunctor(\star, a) \).
Let \( \beta \in \Grp(a,b) \).
Then \( \beta^{-1} \cdot \Group \cdot \beta = \Stabiliser(x \cdot \beta) \).

\quad

(2, \(\Leftarrow\))

Assume that \( K \) satisfies the condition.
For each \( a \in \Grp \) and \( \Group \in K_a \), we show that \( \Group \in (\KitOf{\InducedProf{\Group}})_a \) and \( (\KitOf{\InducedProf{\Group}})_b \subseteq K_b \) for every \( b \in \Grp \).
Then \( \KitOf{(\coprod_{a, \Group \in K_a} \InducedProf{\Group})} = K \).

Trivially, \( \Group \in (\KitOf{\InducedProf{\Group}})_a \) holds.

Assume that \( \Groupb \in (\KitOf{\InducedProf{\Group}})_b \).
Then \( \Groupb = \Stabiliser(x) \) for some \( x \in \InducedProf{\Group}(\star, b) \).
By definition, \( x = \Group \cdot \beta \) for some \( \beta \in \Grp(a,b) \), so
\begin{align*}
    \Groupb
    &= \{\, \gamma \in \Grp(b,b) \mid \Group \cdot \beta = \Group \cdot \beta \cdot \gamma \,\}
    \\
    &= \{\, \gamma \in \Grp(b,b) \mid \Group = \Group \cdot \beta \cdot \gamma \cdot \beta^{-1} \,\}
    \\
    &= \{\, \gamma \in \Grp(b,b) \mid (\beta;\gamma;\beta^{-1}) \in \Group \,\}
    \\
    &= \beta^{-1} \cdot \Group \cdot \beta.
\end{align*}
Hence, by the condition on \( K \), we have \( \Groupb \in K_b \).

 \section{Proof of \cref{lem:definability:loader-linking}}\label{sec:appx:loader}
Since \( \Tot \hookrightarrow \TotProf \), the family \( \varpi \) contains a logical family parameterized by totality spaces, so the claim is essentially proved by Loader~\cite[Lemma~8]{Loader1994}.
A subtlety here is that the logicality of this paper is slightly weaker than the uniformity by Loader (see \cite[Definition~2]{Loader1994} for the definition).
However, careful examination reveals that the right-to-left direction is not needed to prove \cite[Lemma~8]{Loader1994}.
Here, we give a self-contained proof.

In order to distinguish occurrences of the same propositional variable in \( \Formula \), we introduce another formula \( \dot{\Formula} \).
The formula \( \dot{\Formula} \) is essentially the same as \( \Formula \) except that all propositional variables are renamed and have exactly one occurrence.
For example, for \( \Formula = (\PropVar_2 \otimes \PropVar_1^{\bot}) \llpar \PropVar_2^{\bot} \), we can choose \( \dot{\Formula}(\PropVarb_1,\PropVarb_2,\PropVarb_3) = (\PropVarb_1 \otimes \PropVarb_2^{\bot}) \llpar \PropVarb_3^{\bot} \).
Then \( \Formula = \dot{\Formula}(\PropVar_2, \PropVar_1, \PropVar_2) \).

Let \( \PropVar_1,\dots,\PropVar_k \) be the list of propositional variables in \( \Formula \) and \( \PropVarb_1,\dots,\PropVarb_n \) be the list of propositional variables in \( \dot{\Formula} \).
Without loss of generality, we may assume that the variables \( \PropVarb_1,\dots,\PropVarb_n \) occur in the formula \( \dot{\Formula} = \dot{\Formula}(\PropVarb_1,\dots,\PropVarb_n) \) in this order.
Let \( (\xi_i)_{i = 1,\dots,n} \) be the sequence of numbers such that \( \Formula = \dot{\Formula}(X_{\xi_1}, \dots, X_{\xi_n}) \).

An \emph{occurrence} is a number \( j \in \{ 1,\dots,n \} \).
It is negative when \( \PropVarb_j^{\bot} \) appears in \( \dot{\Formula} \) and positive if it is not negative. 

\begin{lemma}
    Let \( \vec{B} \) be a sequence of totality spaces and \( (x_1,\dots,x_n) \in \Truncation{\varpi_{\vec{B}}} \) be an element.
    For every positive occurrence \( a \),
there exists a negative occurrence \( b \)
such that \( \xi_a = \xi_b \) (i.e., \( a \) and \( b \) are occurrences of the same propositional variable in \( \Formula \)) and \( x_a = x_b \).
    The same holds with the roles of positive and negative occurrences interchanged.
\end{lemma}
\begin{proof}
    We prove by contradiction.
    Assume that there is no negative occurrence that satisfies the requirement.
    We can assume without loss of generality that
\( i = 1 \).

    We define totality spaces \( C_1,C_2,\dots,C_k \) and functions \( f_i, f'_i \colon \Truncation{B_i} \to \Truncation{C_i} \) as follows (where \( C_i = (\Truncation{C_i}, \TotalElement(C_i), \CototalElement(C_i)) \)).
    \begin{itemize}
        \item \( \Truncation{C_1} \defeq \Truncation{B_1} \uplus \{ \ast \} \).  Here \( \ast \) is an element not in \( \Truncation{B_1} \).
        \item \( C_i \defeq B_i \) for every \( 2 \le i \le k \).
        \item \( f_i(x) \defeq x \) for every \( 1 \le i \le k \) and \( x \in \Truncation{B_i} \).
        \item \( f'_i(x) \defeq x \) if \( 1 < i \le k \) or \( x \neq x_a \).  \( f'_1(x_a) \defeq \ast \).
        \item \( \TotalElement(C_1) \defeq \{ f_i(s) \mid s \in \TotalElement(B_i) \} \cup \{ f'_i(s) \mid s \in \TotalElement(B_i) \} \).
        \item \( \CototalElement(C_1) \defeq \{ f_i(s) \cup f'_i(s) \mid s \in \CototalElement(B_i) \} \).
    \end{itemize}
    Since \( \varpi \) is a logical family, \( f(x_1,\dots,x_n) \defeq (f_{\xi_1}(x_1), \dots, f_{\xi_n}(x_n)) \) and \( f'(x_1,\dots,x_n) \defeq (f'_{\xi_1}(x_1), \dots, f'_{\xi_n}(x_n)) \) belong to \( \Truncation{\varpi_{\vec{B}}} \).

    For a negative occurrence \( i \), let \( t_i \) be a total subset \( t_i \in \TotalElement(C_{\xi_i}) \) such that \( x_i \in t_i \).
    Since \( x_i \neq x_a \) for every negative occurrence \( i \) with \( \xi_i = 1 \), we have \( f'_{\xi_i}(x_i) = x_i \in t_i \) for every negative occurrence \( i \).

    For a positive occurrence \( i \), we define a cototal subset \( t_i \in \CototalElement(C_{\xi_i}) \) let \( t_i \) as follows.
    \begin{itemize}
        \item If \( i \) is an occurrence of \( B_j \) (i.e., \( \xi_i = j \)) with \( j \neq 1 \), let \( t_i \in \CototalElement(C_{\xi_i}) \) be any cototal subset such that \( x_i \in t_i \).
        \item If \( i \) is an occurrence of \( B_1 \) (i.e., \( \xi_i = 1 \)), let \( t_i = f_0(s) \cup f'_0(s) \) for some cototal subset \( s \in \CototalElement(B_1) \) with \( x_i \in s_i \).
    \end{itemize}
    Note that \( f_i(x_i), f'_i(x_i) \in t_i \) for every positive occurrence \( i \).

    Now \( t \defeq t_1 \times \dots \times t_n \in \CototalElement(A(\vec{B})) \).
    We have \( f(\vec{x}), f'(\vec{x}) \in \Truncation{\varpi_{\vec{C}}} \cap t \).
    Since \( f(\vec{x}) \neq f'(\vec{x}) \), we have \( \varpi_{\vec{C}} \notin \TotalElement(A(\vec{B})) \).
\end{proof}

\begin{lemma}\label{lem:appx:totality:inclusion}
    There exists pair of surjections \( \phi_+ \) (from negative occurrences to positive occurrences) and \( \phi_- \) (from positive occurrences to negative occurrences) such that \( i \) and \( \phi_+(i) \) (resp.~\( i \) and \( \phi_-(i) \)) are occurrences of the same formula and both
    \begin{equation*}
        \{ (x_1,\dots,x_n) \mid x_{\phi_+(i)} = x_i \mbox{ for every positive \( i \)} \}
    \end{equation*}
    and
    \begin{equation*}
        \{ (x_1,\dots,x_n) \mid x_{\phi_-(i)} = x_i \mbox{ for every negative \( i \)} \}
    \end{equation*}
    are contained in \( \varpi_{\vec{B}} \) for every \( \vec{B} \).
\end{lemma}
\begin{proof}
    We construct \( \phi_- \).  The construction of \( \phi_+ \) is similar.
    Assume that \( A \) has \( k \) propositional variables.

    For each propositional variable \( \alpha \), let \( N_\alpha \) be the discrete totality space over the negative occurrences of \( \alpha \).

    For a positive \( i \), let \( t_i \defeq N_\alpha \in \CototalElement(N_\alpha) \).
    For a negative \( i \), let \( t_i \defeq \{ i \} \in \TotalElement(N_{\alpha}) \).
    Then \( t_1 \times \dots \times t_n \in \CototalElement(A(N_{\alpha_1}, \dots, N_{\alpha_k})) \).
    So there exists a unique element in the intersection between \( t_1 \times \dots \times t_n \) and \( \varpi_{\vec{N}} \).
    Assume that \( (a_1,\dots,a_n) \) is the unique element.
    Then \( \phi_- \) is defined by \( \phi_-(i) \defeq a_i \).

    Take any totality spaces \( \vec{B} = (B_1,\dots,B_k) \) and \( x_i \) be an element of \( B_{\xi_i} \).
    Assume that \( x_i = x_{\phi_i(i)} \) for every positive \( i \).
    Let \( f_j \colon \Truncation{N_j} \to |B_j| \) be a function defined by \( f_j(i) = x_i \) (for every \( i \in \Truncation{N_j} \)).
    Then, for positive \( i \),
    \begin{equation*}
        f_j(a_i) = f(\phi_-(i)) = x_{\psi_-(i)} = x_i
    \end{equation*}
    and, for negative \( i \) (since \( a_i = i \) for a negative \( i \)),
    \begin{equation*}
        f_j(a_i) = f(i) = x_i.
    \end{equation*}
    This means that \( f(\vec{a}) = (x_1,\dots,x_n) \).
    By the logicality of \( \varpi \), we have \( (x_1,\dots,x_n) \in \varpi_{\vec{B}} \).
\end{proof}

\begin{lemma}\label{lem:appx:totality:equality-condition}
    Let \( \vec{A} \) be totality spaces.
    \begin{enumerate}
        \item Suppose \( (x_1,\dots,x_m), (y_1,\dots,y_m) \in \Truncation{\varpi_{\vec{A}}} \).  Then \( x_i = y_i \) for every positive \( i \) if and only if \( x_i = y_i \) for every negative \( i \).
        \item If \( (x_1,\dots,x_m) \in \Truncation{\varpi_{\vec{A}}} \), then \( x_i = x_{\phi_{+}(i)} \) for every negative \( i \) and \( x_i = x_{\phi_{-}(i)} \) for every positive \( i \).
    \end{enumerate}
\end{lemma}
\begin{proof}
    $(1)$
    Assume \( (x_1,\dots,x_m), (y_1,\dots,y_m) \in \Truncation{\varpi_{\vec{A}}} \) and that \( x_i = y_i \) for every negative \( i \).
    We show that \( x_i = y_i \) for every positive \( i \).

    Let \( B_j \) be the discrete space over \( \Truncation{A_j} \) (i.e., \( \TotalElement(B_j) \) is the set of singletons).
    By the logicality, we have \( \Truncation{\varpi_{\vec{A}}} = \Truncation{\varpi_{\vec{B}}} \).
    Hence, \( \vec{x}, \vec{y} \in \Truncation{\varpi_{\vec{B}}} \).
    For positive \( i \), let \( t_i = \Truncation{B_{\xi_i}} \in \CototalElement(B_{\xi_i}) \); for negative \( i \), let \( t_i = \{ x_i \} = \{ y_i \} \in \TotalElement(B_{\xi_i}) \).
    Then \( t \defeq t_1 \times \dots \times t_m \in \CototalElement(A(\vec{B})) \), so \( (\varpi_{\vec{B}}) \cap t \) is a singleton.
    Since \( (\vec{x}), (\vec{y}) \in ((\varpi_{\vec{B}}) \cap t) \), we have \( \vec{x} = \vec{y} \).

    $(2)$
    Take \( (x_1, \dots, x_m) \in \Truncation{\varpi_{\vec{A}}} \).
    Let \( \vec{y} = (y_1,\dots,y_m) \) be the tuple given by \( y_i = x_i \) for negative \( i \) and \( y_i = x_{\phi_{-}(i)} \) for positive \( i \).
    By \cref{lem:appx:totality:inclusion}, we have \( \vec{y} \in \Truncation{\varpi_{\vec{A}}} \).
    By $(1)$, we have \( x_i = y_i \) for every \( i \).
    So \( x_i = x_{\phi_{-}(i)} \) for every positive \( i \).
\end{proof}

\begin{corollary}
    For any totality spaces \( \vec{A} \),
    \begin{align*}
        \Truncation{\varpi_{\vec{A}}}
        &= \{ (x_1,\dots,x_n) \mid x_i = x_{\phi_{+}(i)} \mbox{ for every negative \( i \)} \} \\
        &= \{ (x_1,\dots,x_n) \mid x_i = x_{\phi_{-}(i)} \mbox{ for every positive \( i \)} \}.
    \end{align*}
    Furthermore, \( \phi_+ \) is the inverse of \( \phi_{-} \).
\end{corollary}
\begin{proof}
    \Cref{lem:appx:totality:inclusion} proved the inclusion in the \( ({\supseteq}) \)-direction.
    The other direction follows from \cref{lem:appx:totality:equality-condition}.

    For the second part, assume for the contradiction that \( \phi_{+} \) is not an injection.
    Then, there exists negative indexes \( i \neq i' \) such that \( \phi_{+}(i) = \phi_{+}(i') \).
    Let \( \vec{A} \) be a sequence of totality spaces with at least two elements.
    Take \( x_j \) for negative \( j \) so that \( x_i \neq x_{i'} \).
    We define \( x_j \) for a positive \( j \) by \( x_j \defeq x_{\phi_{-}(j)} \).
    Then \( \vec{x} = (x_1,\dots,x_n) \in \Truncation{\varpi_{\vec{A}}} \).
    So \( x_j = x_{\phi_{+}(j)} \) for every negative \( j \).
    Then \( x_{i} = x_{\phi_{+}(i)} = x_{\phi_{+}(i')} = x_{i'} \), a contradiction.
\end{proof}

This corollary shows that \( \varpi \) is the interpretation of the proof structure induced by the involution \( [\phi_+, \phi_-] \).

The above argument, given by Loader~\cite{Loader1994}, shows that the desired equality when \( \vec{\Grp} \) are sets, say \( \vec{\Grp} = \Truncation{\vec{\Grpb}} \).
To prove the equality for all groupoids, note that \( \Set \hookrightarrow \GroupoidCat \) is reflective and \( \Grp \longrightarrow \Truncation{\Grp} \) is essentially bijective on objects.
So \( \Truncation{\varpi_{\vec{\Grp}}} = \Truncation{\varpi_{\Truncation{\vec{\Grp}}}} = \Truncation{\sem{\pi}_{\Truncation{\vec{\Grp}}}} = \Truncation{\sem{\pi}_{\vec{\Grp}}} \).

\end{document}